\documentclass[letterpaper, 10 pt, conference]{ieeeconf}  
\IEEEoverridecommandlockouts  % This command is only needed if you want to use the \thanks command

\usepackage{fix-cm}
\usepackage{etex}

\usepackage{dblfloatfix}

\usepackage{nag}

\makeatletter
\@ifpackageloaded{xcolor}{}{%
\usepackage[table,x11names,dvipsnames,svgnames]{xcolor}%
}
\makeatother

\usepackage{colortbl}

\usepackage{graphicx}
\usepackage{wrapfig}

\definecolorset{RGB}{lyft}{}{Red,194,39,36;Sunset,202,53,33;Orange,205,68,20;Amber,200,117,42;Yellow,242,169,52;Citron,186,188,44;Lime,112,159,33;Green,56,139,31;Mint,45,118,56;Teal,52,133,135;Cyan,60,132,202;Blue,55,94,248;Indigo,64,13,247;Purple,115,42,248;Pink,176,25,145;Rose,176,32,75}

\usepackage{cite}

\usepackage{microtype}

\usepackage{array}
\usepackage{multirow}
\usepackage{booktabs}
\usepackage{makecell} %introduces \thead and \makecell

\newcommand{\myparagraph}[1]{\textbf{\emph{#1}}.}
\ifcsname labelindent\endcsname
\let\labelindent\relax
\fi
\usepackage[inline]{enumitem}

\usepackage{subfig}

\newenvironment{lenumerate}[2][]
{\begin{enumerate}[label=(#2\arabic*),leftmargin=0.2in,itemindent=0.15in,#1]}
{\end{enumerate}}

\setlist*[enumerate,1]{label={\itshape\arabic*)}}

\makeatletter
\newcommand{\paragraphswithstop}{%
\let\copyparagraph\paragraph%
\renewcommand\paragraph[1]{\copyparagraph{##1.}}%
}
\makeatother

\usepackage[framemethod=tikz]{mdframed}

\usepackage{suffix}

\usepackage{environ}

\makeatletter
\newsavebox{\boxifnotempty}
\newcommand{\displayifnotempty}[3]{\sbox\boxifnotempty{#2}\setbox0=\hbox{\usebox{\boxifnotempty}\unskip}%
\ifdim\wd0=0pt
\else
 #1\usebox{\boxifnotempty}#3%
\fi%
}
\newcommand{\ifempty}[2]{\setbox0=\hbox{#1\unskip}%
\ifdim\wd0=0pt%
 #2%
\fi%
}
\newcommand{\ifnotempty}[2]{\setbox0=\hbox{#1\unskip}%
\ifdim\wd0>0pt%
 #2%
\fi%
}
\makeatother

\usepackage{algpseudocode}
\usepackage{algorithm}

\usepackage{scrlfile}

\makeatletter
\newcommand*\newstoreddef[1]{
  \BeforeClosingMainAux{%
    \immediate\write\@auxout{%
      \string\restoredef{#1}{\csname #1\endcsname}%
    }%
  }%
}
\newcommand*{\restoredef}[2]{% used at the aux file
  \expandafter\gdef\csname stored@#1\endcsname{#2}%
}
\newcommand*{\storeddef}[1]{
  \@ifundefined{stored@#1}{0}{\csname stored@#1\endcsname}%
}
\makeatother

\usepackage{pageslts}
\NewEnviron{tee}{\BODY\typeout{Marker Tee [start] ^^J \BODY ^^JMaker Tee [end]}}

\input{preamble/math}
\newcommand{\real}[1]{\mathbb{R}^{#1}{}}

\newcommand{\bmat}[1]{\begin{bmatrix}#1\end{bmatrix}}

\newcommand{\transpose}{^\mathrm{T}}

\newcommand{\inverse}{^{-1}}

\DeclarePairedDelimiter{\abs}{\lvert}{\rvert}

\DeclarePairedDelimiter{\norm}{\lVert}{\rVert}

\newcommand{\vct}[1]{\mathbf{#1}}

\DeclareMathOperator{\stack}{stack}

\newcommand{\subjectto}{\textrm{subject to }}

\providecommand{\cC}{\mathcal{C}}

\providecommand{\cE}{\mathcal{E}}

\providecommand{\cG}{\mathcal{G}}

\providecommand{\cI}{\mathcal{I}}

\providecommand{\cL}{\mathcal{L}}

\providecommand{\cO}{\mathcal{O}}
\providecommand{\cP}{\mathcal{P}}

\providecommand{\cS}{\mathcal{S}}
\providecommand{\cT}{\mathcal{T}}
\providecommand{\cU}{\mathcal{U}}
\providecommand{\cV}{\mathcal{V}}

\providecommand{\cX}{\mathcal{X}}
\providecommand{\cY}{\mathcal{Y}}

\usepackage{units}

\newcommand{\newcolorlabel}[2]{%
  \expandafter\newcommand\csname #1\endcsname[1]{%
    \colorbox{#2}{\color{white}\textsf{\textbf{##1}}}}%
}

\newcommand{\newcommenter}[2]{%
  \expandafter\newcommand\csname #1\endcsname[1]{%
    \fcolorbox{#2}{#2}{\color{white}\textsf{\textbf{#1}}}
    {\color{#2}##1}}%
  \expandafter\newcommand\csname at#1\endcsname{%
    \fcolorbox{#2}{#2}{\color{white}\textsf{\textbf{@#1}}}
    {\color{#2}}}%
  \expandafter\newcommand\csname #1hl\endcsname[2]{%
    \colorbox{#2}{\color{white}\textsf{\textbf{#1}}}\sethlcolor{Azure2}\hl{##2}~%
    \expandafter\ifx\csname commentarrow\endcsname\relax$\leftarrow$\else \commentarrow[#2]\fi~%
    {\color{#2}##1}}%
  \expandafter\newcommand\csname #1st\endcsname[2]{%
    \colorbox{#2}{\color{white}\textsf{\textbf{#1}}}\sout{##2}~%
    \expandafter\ifx\csname commentarrow\endcsname\relax$\leftarrow$\else \commentarrow[#2]\fi~%
    {\color{#2}##1}}%
}
\newcommenter{TODO}{DodgerBlue1}
\newcommenter{rtron}{Green3}

\usepackage{comment}

\usepackage{pdfcomment}

\usepackage{soul}

\usepackage[normalem]{ulem}

\usepackage{csquotes}

\usepackage{tikz}
\usetikzlibrary{calc}
\usetikzlibrary{matrix}
\usetikzlibrary{chains}
\usetikzlibrary{shapes.geometric}
\usetikzlibrary{arrows.meta}
\usetikzlibrary{decorations.pathreplacing}
\usetikzlibrary{backgrounds}

\tikzset{
  dim above/.style={to path={\pgfextra{
        \pgfinterruptpath
        \draw[>=latex,|->|] let
        \p1=($(\tikztostart)!1.5em!90:(\tikztotarget)$),
        \p2=($(\tikztotarget)!1.5em!-90:(\tikztostart)$)
        in(\p1) -- (\p2) node[pos=.5,sloped,above]{#1};
        \endpgfinterruptpath
      }
    }
  },
  dim double above/.style={to path={\pgfextra{
        \pgfinterruptpath
        \draw[>=latex,|->|] let
        \p1=($(\tikztostart)!3em!90:(\tikztotarget)$),
        \p2=($(\tikztotarget)!3em!-90:(\tikztostart)$)
        in(\p1) -- (\p2) node[pos=.5,sloped,above]{#1};
        \endpgfinterruptpath
      }
    }
  },
  dim below/.style={to path={\pgfextra{
        \pgfinterruptpath
        \draw[>=latex,|->|] let 
        \p1=($(\tikztostart)!-1em!-90:(\tikztotarget)$),
        \p2=($(\tikztotarget)!-1em!90:(\tikztostart)$)
        in (\p1) -- (\p2) node[pos=.5,sloped,below]{#1};
        \endpgfinterruptpath
      }
    }
  },
}

\tikzset{
    right angle quadrant/.code={
        \pgfmathsetmacro\quadranta{{1,1,-1,-1}[#1-1]}     % Arrays for selecting quadrant
        \pgfmathsetmacro\quadrantb{{1,-1,-1,1}[#1-1]}},
    right angle quadrant=1, % Make sure it is set, even if not called explicitly
    right angle length/.code={\def\rightanglelength{#1}},   % Length of symbol
    right angle length=2ex, % Make sure it is set...
    right angle symbol/.style n args={3}{
        insert path={
            let \p0 = ($(#1)!(#3)!(#2)$) in     % Intersection
                let \p1 = ($(\p0)!\quadranta*\rightanglelength!(#3)$), % Point on base line
                \p2 = ($(\p0)!\quadrantb*\rightanglelength!(#2)$) in % Point on perpendicular line
                let \p3 = ($(\p1)+(\p2)-(\p0)$) in  % Corner point of symbol
            (\p1) -- (\p3) -- (\p2)
        }
    }
}

\newcommand{\pgfextractangle}[3]{%
    \pgfmathanglebetweenpoints{\pgfpointanchor{#2}{center}}
                              {\pgfpointanchor{#3}{center}}
    \global\let#1\pgfmathresult  
}

\usetikzlibrary{shapes.arrows}
\newcommand{\commentarrow}[1][Azure4]{\tikz[baseline=-3pt]{\node[shape border uses incircle, fill=#1,rotate=180,single arrow, inner sep=1pt, minimum size=6pt, single arrow head extend=2pt]{};}}

\tikzset{ax/.style={-latex,line width=2pt}}

\tikzset{camera/.style={fill=Sienna1,fill opacity=0.5},%
image plane/.style={draw=RoyalBlue3,line width=2pt}}

\newcommenter{todo}{Firebrick1}

\usepackage{amsmath}
\usepackage{bm}
\usepackage{graphicx}
\usepackage{algpseudocode}
\usepackage{algorithm}
\DeclareGraphicsExtensions{.pdf,.png,.jpg}
\usepackage{alphabeta}
\usepackage{fixltx2e}

\newcommand{\rrtstar}{{\texttt{RRT$^*$}}}

\def\BibTeX{{\rm B\kern-.05em{\sc i\kern-.025em b}\kern-.08em
    T\kern-.1667em\lower.7ex\hbox{E}\kern-.125emX}}
\begin{document}
\title{Designing Robust Linear Output Feedback Controller based on CLF-CBF framework via Linear~Programming(LP-CLF-CBF)}
%%%Offline synthesizing of robust linear controllers in Polyhedral Environments based on CLF-CBF framework via linear programming
\author{Mahroo Bahreinian$^{1}$, Mehdi Kermanshah $^{2}$, Roberto Tron$^{3}$
  \thanks{This work was supported by ONR MURI N00014-19-1-2571 ``Neuro-Autonomy: Neuroscience-Inspired Perception, Navigation, and Spatial Awareness''}% <-this % stops a space
  \thanks{$^{1}$Mahroo Bahreinian is with Division of Systems Engineering at Boston University, Boston, MA, 02215 USA. Email:
    {\tt\small mahroobh@bu.edu}}%
\thanks{$^{2}$Mehdi Kermanshah is with Department of Mechanical at Boston University, Boston, MA, 02215 USA. Email:
    {\tt\small mker@bu.edu}}%
  \thanks{$^{3}$Roberto Tron is with Faculty of Department of Mechanical Engineering at Boston University, Boston, MA, 02215 USA. Email:
    {\tt\small tron@bu.edu}}}

\maketitle

\begin{abstract}
We consider the problem of designing output feedback controllers that use measurements from a set of landmarks to navigate through a cell-decomposable environment using duality, Control Lyapunov and Barrier Functions (CLF, CBF), and Linear Programming. We propose two objectives for navigating in an environment, one to traverse the environment by making loops and one by converging to a stabilization point while smoothing the transition between consecutive cells. We test our algorithms in a simulation environment, evaluating the robustness of the approach to practical conditions, such as bearing-only measurements, and measurements acquired with a camera with a limited field of view.

\end{abstract}
%%%%%%%%%%%%%%%%%%%%%%%%%%%%%%%%%%%%%%%%%%%%%%%%%%%%%%%%%%%%%%%%%%%%%%%%%%%%%%%%
\section{INTRODUCTION}

Path planning is a major research domain within mobile robotics, involved primarily with the finding of a nominal trajectory from an initial state to a goal, ensuring collision avoidance. Classical path planning methods focus on finding a  \textit{single, nominal paths} within a static and \textit{pre-known map}. These algorithms often assume that the robotic agent is equipped with a lower-level \textit{state feedback} controller, which enables tracking the nominal path despite the presence of extrinsic perturbations and inaccuracies in the model.
In contrast, biological systems demonstrate a more flexible approach. Take, for example, a person navigating through an unfamiliar room: despite the absence of a detailed layout of the space and precise self-localization, the individual can navigate with remarkable reliability and robustness toward a desired exit. This capability in biological systems stems from complex processes that are yet to be fully understood. 

In this paper, we aim to bridge the gap between algorithmic path planning and the inherent capabilities observed in biological systems. We propose the synthesis of \textit{output-feedback controllers} that is robust to \textit{inexact map awareness}. By focusing on controller synthesis rather than fixed-path generation, we integrate the high-level path planning with the low-level control processes. 
% This integration leads to a more robust controller to disturbances and imperfect mapping without necessitating re-planning.
Furthermore, the focus on controller-based planning allows for the direct utilization of measurements available to the agent, instead of assuming full state knowledge; finally, since the controllers depend on the environment indirectly (through measurements that are taken online), we empirically show that such controllers are robust to (often very significant) changes in the map. In order to pursue strong theoretical guarantees, in this paper, we assume agents with controllable linear dynamics, and environments that admit a polygonal convex cell decomposition (e.g., via Delaunay triangulations \cite{fortune1992voronoi} or trapezoidal decompositions \cite{latombe2012robot}). Methods to address these limitations are planned as part of our future work (see also the Conclusions section).

\myparagraph{Related works}
Existing works on path planning can be roughly classified into two categories: combinatorial path planning methods, and sample-based path planning methods \cite{comparative}.
Some of the path planning methods consider a continuous model for the environment and therefore provide a continuous path, such as potential fields \cite{khatib1986real}, \cite{krogh1984generalized} and navigation functions \cite{rimon1992exact}, while the other group solves the planning problem by abstracting the environment to a finite representation and find a discrete path, such as probabilistic roadmaps \cite{kavraki1996probabilistic} and cell decomposition methods \cite{lingelbach2004path}.

One of the well-known combinatorial path planning algorithms is cell decomposition, where a complex environment is decomposed into a set of cells, avoiding obstacles by planning straight paths in individual cells; for each individual step, traditional methods use midpoints \cite{lavalle2006planning, schurmann2009computational,choset2005principles}, while more recent solutions aim to optimize path length \cite{kloetzer2015optimizing}. Our work can be seen as a descendant of previous work that handles the cell decomposition vis-\'a-vis the continuous dynamic through a hybrid system perspective by synthesizing a state-feedback controller for each cell. Initial work proposed potential-based controllers \cite{conner2003composition}, while others characterize the theoretical conditions \cite{habets2006reachability}
%\rtron{Cite Habet et al., ``Reachability  and  Control  Synthesis  forPiecewise-Affine Hybrid Systems on Simplices'', Roszak et al., ``Necessary and Sufficient Conditions for Reachability on a Simplex''}
and closed-form solutions \cite{belta2005discrete} for linear affine controllers. Although the latter approaches were extended to nonlinear systems in \cite{girard2008motion} and uncertain maps \cite{yan2008mobile} (using intelligent re-planning), they all assume that each cell in the decomposition is a \emph{simplex} (a polytope in $\real{d}$ with $d+1$ vertices, e.g., a 2-D triangle). In contrast, our method can handle arbitrary convex polytopes, and design \emph{output}-feedback controllers (instead of state-feedback). In this paper, we only consider 'reach-avoid' problems. However, our approach can be extended for broader spatial-temporal Logic specifications such as linear temporal logic (LTL).\cite{wu2009synthesis, kloetzer2008fully, cohen2021model}

Sampling-based planning algorithms, such as rapidly exploring random trees (\texttt{RRT}), have become popular in the last few years due to their good practical performance and their probabilistic completeness \cite{lavalle2006planning,lavalle2001randomized,karaman2011sampling}. For trajectory planning that takes into account non-trivial dynamical systems of the robot, kinodynamic \texttt{RRT} \cite{lavalle2001randomized, lavalle2006planning} and closed-loop \texttt{RRT} (\texttt{CL-RRT}, \cite{kuwata2008motion}) and \texttt{CL-RRT\#} grow the tree by sampling control inputs and then propagating forward the nonlinear dynamics (with the optional use of stabilizing controllers and tree rewiring to approach optimality). Further, in this line of work, there has been a relatively smaller amount of work on algorithms that focus on producing controllers instead of simple reference trajectories.
The \texttt{safeRRT} algorithm \cite{positiveInvariant,weiss2017motion} generates a closed-loop trajectory from the initial state to the desired goal by expanding a tree of local state-feedback controllers to maximize the volume of corresponding positive invariant sets while satisfying the input and output constraints. %{The algorithm \cite{mcconley2000computationally} expands the region of stability based on a control Lyapunov function to various trim points of the system to construct a control law that guarantees a closed-loop stability}.
Based on the same idea and following the \texttt{RRT} approach, the \texttt{LQR-tree} algorithm \cite{tedrake2009lqr} creates a tree by sampling over state space and stabilizes the tree with a linear quadratic regulator (LQR) feedback. With respect to the present paper, the common trait among all these works is the use of full-state feedback (as opposed to output feedback), although they do not require prior knowledge of the convex cell decomposition of the environment.

Finally, our work builds upon the real-time synthesis of point-wise controls that trade off safety and stability for nonlinear input-affine systems through a Quadratic Program (QP) formulation \cite{ames2014control,hsu2015control}. To the best of our knowledge, our paper is the first to use similar conditions for synthesizing controls over entire convex regions rather than single points.

\myparagraph{Previous work contributions}
A preliminary version of this work was published in \cite{bahreinian2021robust}. In this work, we proposed a novel approach to synthesize a set of output-feedback controllers on a convex cell decomposition of a polygonal environment via Linear Programming (LP). We defined constraints in terms of a Control Lyapunov Function (CLF) and Control Barrier Functions (CBF) to ensure, respectively, stability and safety (collision avoidance) throughout all the states in a cell while automatically balancing the two aspects to maximize robustness. Our formulation results in a linear min-max optimization problem, which is solved by converting it to an LP form. The major contributions of that work are:
\begin{itemize}
\item We allow a cell to be any generic convex polytope (instead of a simplex).
\item We consider output feedback based on any affine function of the state (under the natural assumption that the overall dynamics is controllable), although, for the sake of presenting a concrete application, we focus on controls using measurements of the relative position of the agent with respect to landmarks in the environment.
\item We apply the CLF-CBF to the new framework of control synthesis.
\end{itemize}
\myparagraph{Contributions of this work}
We integrated our solution with the sample-based method in \cite{bahreinian2021rrt} and introduced Gaussian noise to measurements in \cite{wang2021chance}. We extended this approach for probabilistic measurements with bounded uncertainty in \cite{kermanshah2023control}.
In previous works, the environment is decomposed to a set of convex cells; then, the robot drives through cells by switching between controllers. The main contributions of this work are as follows:
\begin{itemize}
    \item Propose a new cost function that smoothens the transition between consecutive cells.
    \item Modify the control synthesize problem to address  cases where the stabilization is in the middle of the cell
    \item Providing theoretical proof for the stability of this modified version
    \item Extending this approach to use only bearing measurements of landmarks
\end{itemize}
After introducing some preliminary definitions (Section~\ref{sec:notation}), we introduce the problem statement and propose our solution (Section~\ref{sec:problem setup}, and then we analyze the stability of the solution mathematically (Section~\ref{sec:stationary}). We conclude the paper with a few illustrative numerical examples (Section~\ref{sec:examples} and Section~\ref{sec:conclusion}).
%%%%%%%%%%%%%%%%%%%%%%%%%%%%%%%%%%%%%%%%%%%%%%%%%%%%%%%%%%%%%%%%%%%%%%%%%%%%%%%%%
\section{NOTATION AND PRELIMINARIES}\label{sec:notation}

In this section, we review CLF and CBF constraints in the context of our application on agents with linear dynamics and a convex cell decomposition of the environment.
%%%%%%%%%%%%%%%%%%%%%%%%%%%%%%%%%%%%%%%%%%%%%%%%%%%%%%%%%%%%%%%%%%%%%%%%%%%%%%%%%
\subsection{System dynamics}
We start by considering a control-affine dynamical system\footnote{The CLF-CBF concepts are applicable to input-affine systems, but in this work, we assume linear time-invariant systems and affine barrier functions.}

where $x \in \cX\subset\real{n}$ denotes the state, $u\in\cU\subset\real{m}{}$ is the system input, and $A\in\real{n\times n}$, $B\in\real{n\times m}$ define the linear dynamics of the system. We assume that the pair $(A,B)$ is controllable, and that  $\cX$ and $\cU$ are polytopic,
\begin{align}\label{state_limits}
  \cX=\{x\mid A_{x}x\leq b_{x}\},&& \cU=\{u\mid A_{u}u\leq b_u\},
\end{align}
and that $0\in\cU$.
We assume that the robot has linear dynamics of the form
\begin{equation}\label{sys1}
  \dot{x}=Ax+Bu,
\end{equation}
In our case, $\cX$ will be a convex cell centered around a sample in the tree (Section~\ref{sec:simplified-tree}).
\begin{definition}
    We divide the state of systems $x$ into two parts $x_{p} \in \mathbb{R}^{n_p}, x_{d} \in \mathbb{R}^{n_d}$ where $x_{p} = P_{p} x$ is the position of the system and $x_{d} =P_{d} x$ the rest of states ($n_p + n_d = n$), where $P_{p} \in \mathbb{R}^{n_p \times n}, P_{d} \in \mathbb{R}^{n_d \times n}$ are orthogonal projection matrices. 
\end{definition}
\begin{definition}
    We only consider constraints decoupled constraints on $x_{p}$ and $x_{d}$. Thus, we can divide $\cX$ into to sets $ \cX_{p}$ and $\cX_{\text{dyn}}$. Where $\cX_{p} = \{x|A_{p}x\leq b_{p}\}$  and $\cX_{\text{dyn}} = \{x|A_{\text{dyn}}x\leq b_{\text{dyn}}\}$ contains all constraints only corresponding for $x_p$ and $x_{d}$ respectively.
\end{definition}
%%%%%%%%%%%%%%%%%%%%%%%%%%%%%%%%%%%%%%%%%%%%%%%%%%%%%%%%%%%%%%%%%%%%%%%%%%%%%%%%%%
\subsection{Control Lyapunov and Barrier Functions (CLF, CBF)}\label{sec:ECBF}
In this section, we review the CLF and CBF constraints, which are differential inequalities that ensure stability and safety (set invariance) of a control signal $u$ with respect to the dynamics \eqref{sys1}. First, it is necessary to review the following.

\begin{definition}
  The Lie derivative of a differentiable function $h$ for the dynamics \eqref{sys1} with respect to the vector field $Ax$ and $B$ is defined as $\cL_{Ax}h(x)=\frac{\partial h(x(t))}{\partial x}\transpose Ax$ and $\cL_B h(X) =\frac{\partial h(x(t))}{\partial x}\transpose B$.
   The Lie derivative of order $r$ is denoted as $\cL_{Ax}^r$, and is recursively defined by $\cL_{Ax}^{r}h(x)=\cL_{Ax}(\cL_{Ax}^{r-1}h(x))$, with $\cL_{Ax}^1h(x)=\cL_{Ax}h(x)$, respectively.
\end{definition}
\begin{definition}
    \label{def:rel_degree}
 A function $h(x)$ has relative degree $r$ with respect to the dynamics \eqref{sys1} if $\cL_B\cL^i_{Ax}h(x)=0$ for all $ i \leq r-1$ and $\cL_B\cL^r_{Ax}h(x) \neq 0$; equivalently, it is the minimum order of the time derivative of the system, $h^r(x)$, that explicitly depends on the inputs $u$.
 Applying this definition to the system \eqref{sys1} we obtain
\begin{equation}\label{Lie_h}
  h^r(x)=\cL_{Ax}^rh(x)+\cL_B\cL_{Ax}^{r-1}h(x)u
  % \dot{h}(x)=\cL_{Ax}h(x)+\cL_Bh(x)u.
\end{equation}

\end{definition}
% \begin{definition}
% Given a function $h(x)$ with relative degree $r$ for the dynamics \eqref{sys1}, we define the transversal state: 
% \begin{equation} \xi_h = \begin{bmatrix}
% h(x) \\ \dot{h}(x) \\ \vdots \\ h^{r-1}(x)
% \end{bmatrix} =  \begin{bmatrix}
% h(x) \\ \cL_{Ax}h(x) \\ \vdots \\ \cL_{Ax}^{r-1}h(x)
% \end{bmatrix}
% \end{equation}
% \end{definition}

% In this work, we assume that Lie derivatives of $h(x)$ of the first order are sufficient \cite{Isidori:book95} (i.e., $h(x)$ has relative degree $1$ with respect to the dynamics \eqref{sys1}); however, the result could be extended to the higher degrees, as discussed in \cite{bahreinian2021robust}.

We now pass on the definition of the differential constraints.

Consider a continuously differentiable function $h(x):\cX\to\real{}$ which defines a safe set $\cC_0$ such that
\begin{equation}\label{set_c}
  \begin{aligned}
    \cC_0&=\{x\in \real{n}|\;h(x)\geq0\},\\
    \partial \cC_0&=\{x\in \real{n}|\;h(x)=0\},\\
    {Int}(\cC_0)&=\{x\in \real{n}|\;h(x)>0\}.
  \end{aligned}
\end{equation}
We say that the set $\cC_0$ is \emph{forward invariant} (also said \emph{positive invariant} \cite{positiveInvariant}) if  $x(t_0) \in \cC_0$ implies $x(t)\in \cC_0$, for all $t\geq 0$~\cite{zcbf1}.
\begin{definition}
    We recursively define function $\psi_i$ as :
    \begin{equation}
    \begin{aligned}
        \psi_0(x) & = h(x)\\
        \psi_1(x) & = \dot{\psi}_0(x) + \alpha_0\psi_0(x) \geq 0\\
        \psi_r(x,u) &= \dot{\psi}_{r-1}(x) + \alpha_{r-1}\psi_{r-1}(x) \geq 0\\
    \end{aligned}
    \end{equation}
Where $\alpha_0\hdots \alpha_{r-1}$ are positive constants. Set $\cC_i$ is defined as $\cC_i = \{x| \psi_i(x) \geq 0 \}$
\end{definition}
\begin{proposition}[HCBF, \cite{hcbf}]\label{def:ECBF} Consider the control system
\eqref{sys1}, and a continuously differentiable function $h(x)$ with relative degree $r \geq 0$ defining a set $\cC_0$ as in \eqref{set_c}. The function
$h(x)$ is a Higher order Control Barrier Function (HCBF) if a control inputs $u \in \cU$ exist such that
  \begin{equation}\label{cons:cbf}
    \psi_r(x,u) \geq 0,\forall x \in \cC_0.
  \end{equation}
  Furthermore, \eqref{cons:cbf} implies that the set $\cC_0 \cap \cC_1 \hdots \cap \cC_r$ is forward invariant.
\end{proposition}
For simpler notation, \eqref{cons:cbf} can be written as:
\begin{equation}
    \psi_r(x) = \cL_{Ax}^rh(x)+\cL_{Ax}^{r-1}\cL_Bh(x)u+c_b\transpose\xi_h(x) \geq 0
\end{equation}
where $\xi_h$ contains all lower order derivative of function $h(x)$
\begin{equation}
\label{eq:eh}\xi_h  =  \begin{bmatrix}
h(x) \\ \cL_{Ax}h(x) \\ \vdots \\ \cL_{Ax}^{r-1}h(x)
\end{bmatrix}  c_b = 
\begin{bmatrix}
\sum_i \alpha_i \\ \sum_{i1,i2} \alpha_{i1} \alpha_{i2} \\ \vdots \\ \sum_{i1 \hdots ir} \alpha_{i1} \hdots \alpha_{ir}
\end{bmatrix}
\end{equation}
and the  $i$-th element of $c_b$ equals to summation of all possible permutation of $\alpha_{i1} \hdots \alpha_{ir}$.

Consider a continuously differentiable function $V(x):\cX\to\real{}$, $V(x)\geq 0$ for all $x\in\cX$, with $V(x)=0$ for some $x\in\cX$.
\begin{proposition}
    \label{def:ECLF}
  The positive definite function $V(x)$ is a \textit{Higher order Control Lyapunov Function} (HCLF) \cite{nguyen2016exponential} with respect to \eqref{sys1} if there exists positive constant vector $c_l$ and control inputs $u\in \cU$ such that
  \begin{equation}\label{cons:clf}
    \begin{aligned}
\cL_{Ax}^{r}V(x)+\cL_{Ax}^{r-1}\cL_BV(x)u+c_l \transpose \xi_V(x)\leq 0,\forall x \in \cX.
    \end{aligned}
  \end{equation}
  Where, $\xi_V(x)$ contains lower order derivative similar to \eqref{eq:eh}. Furthermore, \eqref{cons:clf} implies that $\lim_{t\to\infty}V(x(t))=0$.
\end{proposition}
%-------------------------------------------------------------------------

%%%%%%%%%%%%%%%%%%%%%%%%%%%%%%%%%%%%%%%%%%%%%%%%%%%%%%%%%%%%%%%%%%%%%%%%%%%%%%%%
\subsection{Convex Decomposition of the Environment}\label{sec:simplified-tree}
We start with a tree $\cT = (\cV,\cE)$ generated by the traditional \rrtstar{} algorithm \cite{karaman2011sampling}. Since the number of samples is finite, the generated tree is not optimal, although it has a large number of nodes. We simplify the tree to reduce the number of nodes (while keeping all the samples that are in collision with obstacles) by following the simplified-\rrtstar{} algorithm in \cite{bahreinian2021rrt} and denote it as $\cT=(\cV_s,\cE_s)$.

Note that as a consequence of the simplifying steps above, it is still possible to connect any sample that was discarded from the original \rrtstar{} to the simplified tree with a straight line, suggesting that the simplified tree will be a good road-map representation \cite{choset2005principles} of the free configuration space reachable from the root (up to the effective resolution given by the original sampling).
Given the simplified tree $\cT_s=(\cV_s,\cE_s)$, for each node $i\in \cV_s$ in the tree, we define a convex cell $\cX_{i}$ similar to \cite{bahreinian2021rrt} such that the boundaries of $\cX_{i}$ are defined as the bisectors hyper-plane between node $i$ and other nodes in the tree except node $j$ which node $j$ is the parent of node $i$.
The polyhedron $\cX_{i}$ is similar to a Voronoi region \cite{latombe2012robot}. Note that $\cX_{i}$ contains all the points closer to $i$ than other vertices in $\cT_s$ but also includes the parent~$j$.

We assume the environment $\cP\subset \real{n}$, is decomposed in a finite number of convex cells $\{\cX_{i}\}$, such that $\bigcup_{i} \cX_{i}=\cP$, and set $\cX_{i}$ is a polytope defined by linear inequality constraints of the form $A_{x}\transpose x\leq b_{x}$.

We aim to design a different linear feedback controller $u$ for each cell $\cX_{i}$. The feedback signal used by the controller will be based on linear relative measurements for a set of \emph{landmarks}.
\begin{definition}
  A landmark is a point $l\in \real{n}$ whose location is known and fixed in the environment.
\end{definition}
For each convex section $\cX_{i}$, we have a finite number of landmarks, and the landmarks can be any points in the environment.
%%%%%%%%%%%%%%%%%%%%%%%%%%%%%%%%%%%%%%%%%%%%%%%%%%%%%%%%%%%%%%%%%%%%%%%%%%%%%%%%
\subsection{High-level planning}\label{planning}
We consider two overall objectives for the controller design:
\begin{lenumerate}{O}
\item \label{it:point-stabilization} Point stabilization: given the stabilization point (where $\dot{x}=0$) in the environment and starting from any point, we aim to converge to the stabilization point (e.g., Fig.~\ref{fig:O1}).
\item\label{it:patrolling} Patrolling: starting from any point, we aim to patrol the environment by converging to a path, and then traversing the same path (e.g., Fig.~ \ref{fig:O2}).
\end{lenumerate}

First, we decompose the environment into a set of convex cells by implementing the cell decomposition method in \cite{bahreinian2021rrt} using the sample-based \rrtstar{} method (See Section~\ref{sec:simplified-tree}). To specify the convergence objective for each controller $u$, we first abstract the cell decomposition of the environment into a graph $\cG=(\cV,\cE)$, where each vertex $i \in \cV$ represents a cell $\cX_i$ in the partition of $P$, and an edge $(i,j)\in \cE$ if and only if cells corresponding to $i$ and $j$ have a face in common.

In the case of the point stabilization objective \ref{it:point-stabilization}, the stabilization point is one of the graph's vertices. If the stabilization point is in the middle of the cell, we introduce new constraints to the problem such that it satisfies the point stabilization in the middle of the cell.

For each cell, we then select one \emph{exit edge} (a pointer) such that, when considered together, all such edges provide a solution in the abstract graph $\cG$ to the high-level objective. For instance, in the case of objective \ref{it:point-stabilization}, the exit edge of each cell will point in the direction of the shortest path toward the vertex of the stabilization point. In the case of objective \ref{it:patrolling}, following the exit edges will lead to a cyclic path in the graph.

To give an example, the polygonal environment in Fig.~\ref{g1} is converted to the connected graph in Fig.~\ref{g3} based on the cell decomposition of the environment in Fig.~\ref{g2}. Starting from the first node in Fig.~\ref{g3}, shown by the green point, we find the path from the start node to the equilibrium node indicated by the red point, through the path planning algorithms (e.g., using Dijkstra's algorithm). Regarding that path, we define the \textit{exit face} as the face of the convex section the path moves through, and based on that, we design the controller.
\begin{figure}[t!]
  \centering
  % \rtron{Remove this figure}
  
  \subfloat[]{\label{g1}{\includegraphics[width=3cm]{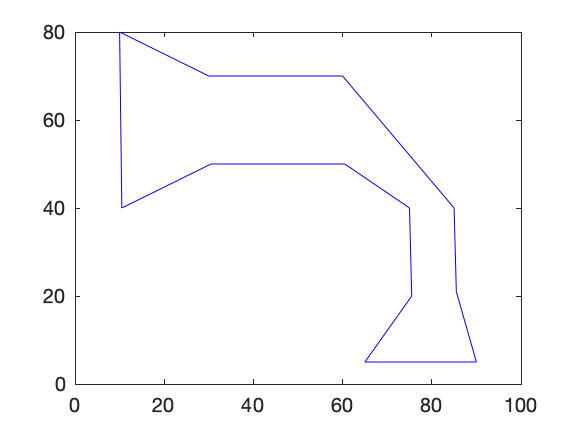} }}%
  \subfloat[]{{\label{g2}\includegraphics[width=3cm]{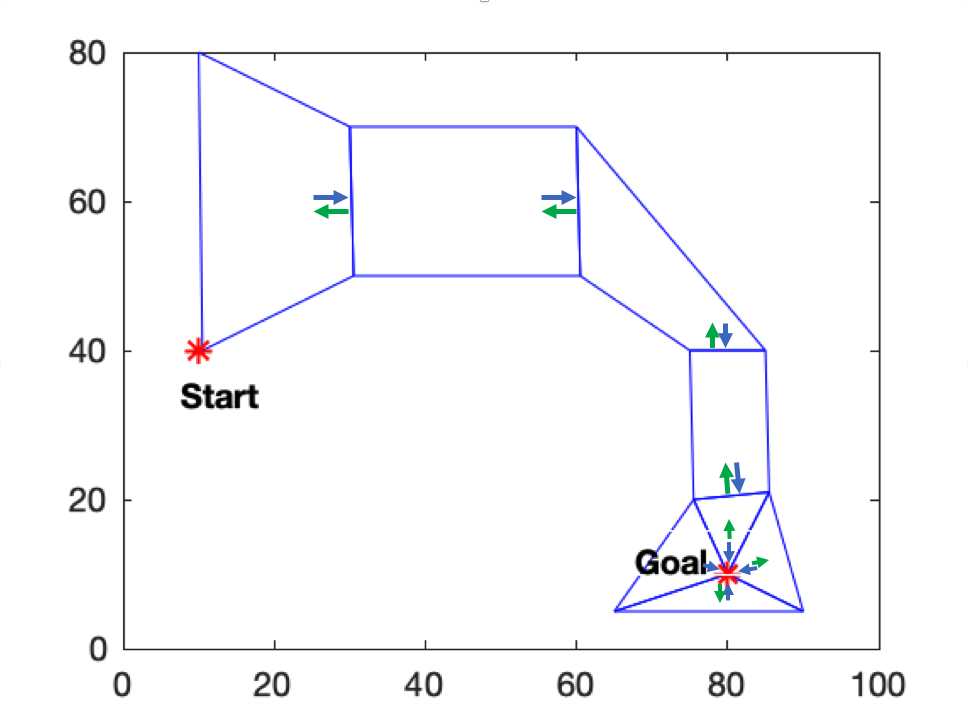} }}%
  \subfloat[]{{\label{g3}\includegraphics[width=3.1cm]{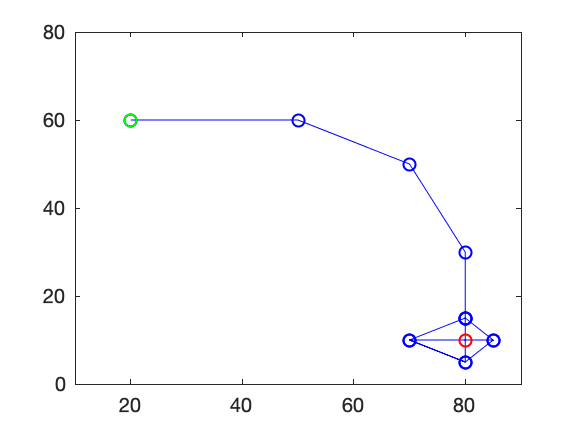} }}%
  \caption{The polygonal environment in Fig~.\ref{g1} is decomposed to 8 convex sections Fig~.\ref{g2}, blue arrows indicate the exit direction and green arrows indicate the inverse exit direction for each cell, the corresponding graph is shown in Fig~. \ref{g3}}
  \label{fig_env}
\end{figure}
\begin{definition}\label{exit_dir}
  For each cell $\cX_i$ in the decomposition of the environment, we define an \emph{exit face} $\cP_{exit}$ to be the face corresponding to the \emph{exit edge} in the abstract graph $\cG$. The \emph{inverse exit direction} $z$ is an inward-facing normal $\cP_{exit}$.
\end{definition}
In this work, we desire to design a controller for each convex section of the environment that drives the system in the exit direction toward the exit face or the stabilization point while avoiding the boundary of the environment.

Overall, thanks to the high-level planning in the abstract graph $\cG$, and the controller design in each cell $\cX_i$ (explained in the sections below), the system will traverse a sequence of cells to reach a given equilibrium point or achieve a periodic steady state behavior (examples in Section~\ref{sec:examples}) according to the desired objective.

%%%%%%%%%%%%%%%%%%%%%%%%%%%%%%%%%%%%%%%%%%%%%%%%%%%%%%%%%%%%%%%%%%%%%%%%%%%%%%%%%
\section{PROBLEM SETUP}\label{sec:problem setup}
This section aims to synthesize a robust controller for a convex cell $\cX$ (with respect to previous sections, we dropped the subscript $i$ to simplify the notation).
We assume that the robot can only measure the relative displacements between the robot's position $x_p$ and the landmarks in the environment, which corresponds to the output function
\begin{equation}\label{vec-landmarks}
  y=(L-x_p\vct{1}\transpose)^\vee=L^\vee-\cI P_px=\stack{(l_i- x_p)},
\end{equation}
where $L\in\real{n\times n_l}$ is a matrix of landmark locations, $i=1,\hdots,n_l$ that $n_l$ is the number of landmarks,  $A^\vee$ represents the vectorized version of a matrix $A$, $\cI=\vct{1}_{nl} \otimes I_n$, and $\otimes$ is the Kronecker product. Our goal is to find a feedback controller for the form
\begin{equation}\label{u}
  u(y,x_{d}) = K_py+K_d x_{d} + K_b
\end{equation}
where $K_p \in \real{m\times nn_l}, K_d \in \real{m\times n_d}$ and $K_b\in\real{m}$  are the feedback gains that need to be found for each cell $\cX$. The goal is to design $u(y,x_d)$ to drive the system toward the exit direction while avoiding obstacles. Note that, to define a controller for a cell, the landmarks do not necessarily need to belong to $\cX$, and, in general, each cell could use a different set of landmarks (see also Section~\ref{sec:limited-field-of-view}).

\begin{remark}
  In general, our framework can handle general linear output $y=Cx+D$, but we focus here on the path planning application.
\end{remark}
\subsection{Control Barrier Function}\label{sec_cbf}
We characterize the safe set $\cC_0$ with $n_h$ linear function $h_i$ as:
\begin{equation}\label{cbf_f}
  h_i(x)={A_{hi}}x+b_{hi}
\end{equation}
Where each $A_{hi}$ is derived from either $A_{dyn}$ or $A_x$ by negating a row while omitting the row corresponding to the exit face.

Note that \eqref{cons:cbf} only enforces the forward invariance if the initial condition lies within the $\bigcap_i^{r-1} \mathcal{C}_i$ Therefore, we assume that $\psi_i(x) \geq 0$ for all $0\leq i\leq r$  are initially met. This assumption leads to the following set of linear constraints:

\begin{equation}
    \psi_i = A_h A^i x + c_b\transpose[0:i] \bmat{A_h x+b_h\\ A_hA x\\ \vdots\\ A_h A^{i-1}x} \label{eq:hcbfsets} 
\end{equation}
 Where $c_b [0:i] \in \mathbb{R}^i$ is a subvector that contains the first $i$ elements of $c_b$. These constraints \eqref{eq:hcbfsets} adds new constraints on dynamic states. In order to incorporate them, we define the restricted dynamical set as $\cX_d = \{x|A_d x\leq b_d \} $ 

% Let $A_{hi}\in\real{1\times n_x}$ belongs to the union of all rows of $A_x$ except the one row defining the exit face, and $b_{hi}\in\real{}$, we define the following candidate CBF as,
% \begin{equation}\label{cbf_f}
%   h_i(x)={A_{hi}}x+b_{hi}
% \end{equation}
% where $i=\{1,\hdots,n_p\}$ such that $n_p$ denotes the number of faces of $\cX$ except the one associated with an exit face (or all of them in the case of a stabilization point).
%%%%%%%%%%%%%%%%%%%%%%%%%%%%%%%%%%%%%%%%%%%%%%%%%%%%%%%%%%%%
\subsection{Control Lyapunov Function}\label{sec_clf}
To stabilize the system, we define the Lyapunov function $V(x)$ for cell $\cX$ as,
\begin{equation}\label{clf_f}
  V(x)=z^T (x-x_e), \;\;
\end{equation}
where $z\in \real n$ is the inverse exit direction for the cell $\cX$ (see Definition~\ref{exit_dir}), and $x_e$ is an arbitrary point belong to the exit face $\cP_{exit}$ (i.e. $P_{d}x_e = 0$). Since this is a linear function that splits the space into two parts such that all $V(x) \geq 0$ for all $x \in \cX_i$  and the function reaches its minimum $V(x)=0$ when $x$ is in the exit face. Note that this Lyapunov function represents the distance $d(x,\cP_{exit})$ between the current system position and the exit face.

\begin{remark}
  The function $V(x)$ can be defined as a function of the vertices of the exit face instead of its normal.
 For instance, in $\real{2}$, we have
   \begin{equation}
    V(x)=\det(\bmat{v_1-v_0 & x-v_0})
  \end{equation}
  where $v_0$ and $v_1$ are two distinct points (e.g., vertices) in the exit face (with their order determining the correct sign in $V(x)$). Based on the same idea, in $\real{3}$, $V(x)=\det(\bmat{v_1-v_0 &v_2-v_0& x_p-v_0})$
  Where $v_0,v_1,v_2$ are three distinct points in the exit face (e.g., three vertices of the exit plane), respectively. This concept can be generalized to any dimension.
\end{remark}
%%%%%%%%%%%%%%%%%%%%%%%%%%%%%%%%%%%%%%%%%%%%%%%%%%%%%%%%%%%%%%%%%%%%%%%%%%%%%%%%%%%%%%%%%%%%
\subsection{Finding the Controller by Robust Optimization}
Our goal is to find controllers $u$ (more precisely, control gains $(K_p, K_d, K_p)$) that maximize the movement of the robot toward the exit face while avoiding the boundary of the environment by satisfying the CLF and CBF constraints respectively.
Using the CLF-CBF constraints reviewed in Section~\ref{sec:notation}, we encode our goal in the following feasibility problem with a heuristic cost to make the transition between cells smoother :
\begin{equation}\label{opt-feasibility}
  \begin{aligned}
    & \min_{\{K_{pi},K_{di}, K_{bi}\}} \sum_{ij} \varphi_{ij}^t+\varphi_{ij}^p\\
    & \textrm{subject to:}\\
    &\bigl\{\text{CBF:}\;-( \cL^r_{A_x}h_{iq}+\cL^{r-1}_{Ax}\cL_Bh_{iq}u+c\transpose_b\xi_{h_{iq}})\bigr\}\leq 0,\\
    & \bigl\{\text{CLF:}\quad\cL_{Ax}^rV_{i}+\cL^{r-1}_{A_x}\cL_BV_{i}u+c_l\transpose\xi_{V_{i}}\bigr\}\leq 0,\\
    &u\in\cU,\;\;\forall x_{p}\in \cX_{i}, \forall x_{{d}} \in \cX_{d}, \forall q\in \mathbb{N}_{n_{hi}}. \forall i \in \mathbb{N}_N
  \end{aligned}
\end{equation}
%\rtron{It should use $x_i$ instead of $x$, because the constraints operate on different regions}
where $(i,j)$ are two consecutive cells and $\mathbb{N}_i$ denotes the natural numbers less or equal than $i$.
Two terms of the objective function in \eqref{opt-feasibility} are defined as,
\begin{equation}\label{defineCost}
  \begin{aligned}
    & \varphi_{ij}^t= \sum_{v^\prime=1}^{n_{v_d}}\sum_{v =1}^{n_v} \abs{K_{1i}(y^v_{ij})+K_{2i} x_{\text{dyn}}^{v^\prime}+K_{3i}\\&-K_{1j}(y^v_{ij})-K_{2j} x_{\text{dyn}}^v-K_{3j} }\\ 
    % &\quad\quad\quad\;\;\abs{K_i(y^2_{ij})+k_i-K_j(y^2_{ij})-k_j},\\
    & \varphi_{ij}^p=  \sum_{v^\prime=1}^{n_{v_d}}\sum_{v =1}^{n_v}\abs{BP_i(K_{1i}(y^v_{ij})+K_{2i}x_{\text{dyn}}^{v^\prime} +K_{3i})-(y^v_{ij}-y_m)}\\
    & y_{ij}^v = (L-x_{ij}^v)\\
    & y_m = \frac{1}{n_v}\sum_{v =1}^{n_v} y^v_{ij}
    %  \abs BP_i(K_1
    % &\quad\quad\quad\;\;\abs{P_i(K_i(y^2_{ij})+k_i)-(x^1_{ij}-x^2_{ij})},\\
    % &y_{ij}^1=(L-x_{ij}^1),\\%\rtron{\textrm{fix superscripts/subscripts}}\\
    % &y_{ij}^2=(L-x_{ij}^2),
    \end{aligned}
\end{equation}
Where $x_{ij}^v$ are the vertices of the common face between two consecutive cells $\cX_i$ and $\cX_j$ and $x_{\text{dyn}}^{v^{\prime}}$ are the vertices of the set $\cX_d$. Basically, $\varphi_{ij}^t$ computes the difference between two controllers on all vertices of the common face and $\cX_d$ and minimizing $\varphi_{ij}^t$ provides a smooth transition between two cells. $\varphi_{ij}^p$ computes the difference between the projection of the controller on the exit face and the vector of the exit face, and minimizing $\varphi_{ij}^p$ results in a controller that drives the agent through the center of the exit face. $P_i \in \real{n\times n}$ defined as
\begin{equation}
     P_i = \cI_2-\eta_iz_i\transpose{z_i},
\end{equation}
where $0\leq\eta_i\leq1$ is a user-defined constant that adjusts the length of the projection of the controller in the exit face.

In practice, we aim to find a controller that satisfies the constraints in \eqref{opt-feasibility} with some margin. For the CLF constraints, we use $\delta_l$ margin to achieve finite-time convergence (see Proposition~\ref{delta_l}), and for the CBF constraints, we use the $\delta_b$ margin to achieve minimum distance from the obstacles (see Proposition~\ref{delta_b}). We focus on the following robust optimization problem:
\begin{equation}\label{opt_margin}
\centering
  \begin{aligned}
    & \min_{K_i,k_i,\delta_l,\delta_b} \sum_{ij} \varphi_{ij}^t+\varphi_{ij}^p+\omega_b \transpose \delta_b+\omega_l \delta_l\\
     & \textrm{subject to:}\\
     &\max_x \text{CBF}\leq \delta_b ,\\
    &\max_x \text{CLF}\leq \delta_l,\\
    & \delta_l,\delta_b \leq 0,\\
    &\forall x\in \cX, \forall x_{dyn} \in \cX_{dyn} \forall u \in \cU
  \end{aligned}
\end{equation}
where weights $\omega_{b}$ and $\omega_{l}$ are user-defined constants defining the trade-off between the barrier and Lyapunov function constraints.
% From \eqref{Lie_h}, the Lie derivatives of $h_{i}(x)$ and $V_{i}(x)$ are written as:
% \begin{equation}\label{cons-Lie}
%   \begin{aligned}
%     & \dot{h}_{i}(x)=A_{h_i}\dot{x} =A_{hi}(A x+B u),\\
%     & \dot{V}_{i}(x) =z_{i}\transpose\dot{x}= z_{i}\transpose (Ax+Bu).
%   \end{aligned}
% \end{equation}
Combining \eqref{vec-landmarks} and \eqref{u} with \eqref{sys1}, the constraints in \eqref{opt_margin} can be rewritten as\\
CBF constraints:
\begin{equation}\label{primal-cbf}
  \begin{aligned}
   % & \text{CBF constraint:}\\
    &\begin{bmatrix}
      &\underset{x}{\max}\gamma_bx\\
      &\subjectto  A_{xi} x \leq b_{xi}\\
       & A_d x \leq b_d 
    \end{bmatrix}
    \leq R_b  
    % \quad \quad  \quad \quad \quad \quad \quad  \quad \quad  \quad \quad
  \end{aligned}
\end{equation}
CLF constraint:
\begin{equation}\label{primal-clf}
  \begin{aligned}
    % & \text{CLF constraint:}\\
    &\begin{bmatrix}\underset{x}{\max} \gamma_v x \\
      \subjectto\;\;A_{xi} x \leq b_{xi}\\
       A_d x \leq b_d    
    \end{bmatrix}
    \leq & R_v
    % \quad \quad  \quad \quad \quad \quad \quad  \quad \quad  \quad \quad\\
  \end{aligned}
\end{equation}
where $\gamma_b$ , $\gamma_v, R_b, R_v$ and $R_u$ equals to:
\begin{equation}
\begin{aligned}
    &\gamma_b = -(A_{hi}A^r+A_{hi}A^{r-1}B (-K_{pi}\cI P_p+ K_{di} P_{d})-\\&{c_b}\transpose \bmat{A_{hi} & A_{hi} A &\hdots & A_{hi}A^{r-1}}) \\
    &\gamma_v = z_{i}\transpose A^{r}+z_{i}\transpose A^{r-1}B(-K_{pi}\cI P_p+K_{di} P_{d})+\\&{c_l}\transpose \bmat{ z\transpose & z\transpose A &\hdots & z\transpose A^{r-1}})\\
     & R_b =  \delta_{h_{i}}+[c_b]_0 b_{hi}+{A_{hi}}A^{r-1}B( K_{pi}L_{i}^\vee+K_{bi})\\
     &R_v =  \delta_{li}-z_{i}^TA^{r-1} B( K_{pi}L_{i}^\vee+K_{bi})+[c_l]_0z\transpose x_e\\
      &R_u = b_u-A_uK_1^i L^\vee-A_u K_{bi}\\
\end{aligned}
\end{equation}
and $[.]_i$ denotes the $i$th element of a vector.
Moreover, the control bounds can be captured as:
\begin{equation}
\label{u bound}
    \begin{aligned}
         &\begin{bmatrix}\underset{x}{\max} A_u( -K_{pi}\cI P_p+ K_{di} P_{d})x \\
      \subjectto\;\;A_{xi} x \leq b_{xi}\\
      A_d x \leq b_d
    \end{bmatrix}
    \leq R_u \\   
    \end{aligned}
\end{equation}
Constraints in \eqref{primal-cbf}, \eqref{primal-clf}  and \eqref{u bound} are linear in terms of variable $x$, so we can write dual forms of the constraints as
CBF dual constraint:
\begin{equation}\label{dual-conscbf}
  \begin{aligned}
    % & \text{CBF dual constraint:}\\
    &\begin{bmatrix}
      &\min_{\lambda_{bi}} \lambda_{bi}\transpose b_{xi} \\
      &\subjectto\\
      & A_{xi} \transpose\lambda_{bi}+A_d\transpose \lambda_{bi}^\prime=\gamma_b\\
      & \lambda_{bi},  \lambda_{bi}^\prime \geq 0,
    \end{bmatrix} \leq R_b 
    % \quad \quad  \quad \quad \quad \quad \quad \quad \quad \quad
  \end{aligned}
\end{equation}
CLF dual constraint:
\begin{equation}\label{dual-consclf}
  \begin{aligned}
    % & \text{CLF dual constraint:}\\
    &\begin{bmatrix}
      &\min_{\lambda_l}\lambda_{li}  \transpose b_{xi}\\
      &\subjectto\\
      &  A_{xi}\transpose \lambda_{li}+A_d\transpose \lambda_{li}^\prime
      =\gamma_v\\
      & \lambda_{li}, \lambda_{li}^\prime   \geq 0
    \end{bmatrix}\leq R_v 
    % \quad \quad  \quad \quad  \quad \quad  \quad  \quad 
  \end{aligned}
\end{equation}
Control bounds dual:
\begin{equation}\label{dual-u bound}
  \begin{aligned}
    % &\text{Control bounds dual:}\\
    &\begin{bmatrix}
      &\min_{\lambda_{ui}}\lambda_{ui}  \transpose b_{xi}\\
      &\subjectto\\
      &  A_{xi}\transpose \lambda_{ui}+A_d\transpose \lambda_{ui}^\prime
      = A_u( -K_{pi}\cI P_p+ K_{di} P_{d})\transpose\\
      & \lambda_{ui}, \lambda_{ui}^\prime \geq 0
    \end{bmatrix}\leq   R_u,
  \end{aligned}
\end{equation}
\begin{equation}\label{opt-dual}
  \begin{aligned}
    \min_{K_{pi},K_{di},K_{bi}, \delta_l,\delta_b} & \sum_{ij} \varphi_{ij}^t+\varphi_{ij}^p+\omega_b \transpose \delta_b+\omega_l \delta_l\\
    \text{subject to} \;\;
      &\min_{\lambda_b} \text{CBF dual constraint }\eqref{dual-conscbf} ,\\
      &\min_{\lambda_l}\text{CLF dual constraint }\eqref{dual-consclf},\\
      &\min_{\lambda_u}\text{Control bounds dual }\eqref{dual-u bound},\\
      &\delta_b,\delta_l \leq 0.
  \end{aligned}
\end{equation}

For the purpose of point stabilization objective\ref{it:point-stabilization}, if the stabilization point $x_e$ is located at the middle of the cell, then, instead of the CLF dual constraint \eqref{opt-dual}, we use the following constraint
\begin{equation}\label{clf_middle}
    K_{pi}(L_i-x_{e}\vct{1}\transpose)+K_{bi}=\vct{0}_2.
\end{equation}

Consequently, \eqref{opt_margin} with the dual constraints becomes:
\begin{equation}\label{opt-min}
  \begin{aligned}
    \min_{K_i,k_i,\delta_b,\lambda_b,\lambda_l}& \sum_{ij} \varphi_{ij}^t+\varphi_{ij}^p+\omega_b \transpose \delta_b\\
    \text{subject to} \;\;
      &\text{CBF dual constraint \eqref{dual-conscbf}},\\
      &\text{Control bounds dual }\eqref{dual-u bound},\\
      &\text{Constraint \eqref{clf_middle} },\\
      &\delta_b \leq 0.
  \end{aligned}
\end{equation}
In the following, we prove that the feasible optimal solution for \eqref{opt_margin} is also the feasible optimal solution for \eqref{opt-dual}.
\begin{remark}\label{strong_duality}
  By strong duality \cite[Theorem 4.4]{LP} if a linear programming problem has an optimal solution, so does its dual, and the respective optimal costs are equal.
\end{remark}
This remark allows us to prove the following.
\begin{lemma}\label{lem1}
  Optimization problems \eqref{opt_margin} and the optimization problem \eqref{opt-dual} have the same feasible optimal solution.
\end{lemma}
\begin{proof}
  Two optimization problems have the same objective functions. Constraints in \eqref{opt_margin} are in the form of LP optimization problem, and the constraints \eqref{opt-dual} are the duals. According to the Remark \ref{strong_duality}, the optimal cost of constraints in \eqref{opt_margin} and \eqref{opt-dual} are equal and result in the same constraints with the same objective functions, which imply the optimization problem \eqref{opt_margin} and \eqref{opt-dual} have the same optimal solution.
\end{proof}
In section \ref{sec:stationary} we will proof that the solution of \eqref{opt-min} is a safe and stable controller for system \eqref{sys1}.
% \begin{lemma}\label{lem2}
%   Optimization problems \eqref{opt-dual} and \eqref{opt-min} have the same feasible optimal solution.
% \end{lemma}
% \begin{proof}
%   Assuming we have an optimal solution for \eqref{opt-min}, then the solution is also feasible for \eqref{opt-dual}, and the objective costs are the same. In the same way, if we have an optimal solution for \eqref{opt-min}, there must exist dual variables for the inner optimization problem in \eqref{opt-dual}, which are also feasible for \eqref{opt-min} and result in the same objective cost\cite{robust_opt}.
% \end{proof}

% %\rtron{Turn this into Theorem}
% \begin{theorem}
% From Lemma \ref{lem1} and Lemma \ref{lem2}, the optimization problems \eqref{opt_margin}, \eqref{opt-dual} and \eqref{opt-min} are equivalent.
% \end{theorem}

%\subsection{Physical meaning of the slack variables $\delta_l$ and $\delta_b$}
In the following two propositions, we study the physical meaning of the slack variables $\delta_l$ and $\delta_b$:
\begin{proposition}\label{delta_l}
For a system with a relative degree order equal to one, if the solution of \eqref{opt-min} results in an optimal $\delta_l$ that is strictly less than zero, then every trajectory exits each cell in finite time.
\end{proposition}
 \begin{proof}
    Define the maximum distance from the exit face as
    \begin{equation}
        d_{max}=\max_{x\in \cX}\{V(x)\}.
    \end{equation}
    For a first-order system, the CLF constraints in \eqref{opt-feasibility} imply
    \begin{equation}
        \dot{V}(x(t))\leq \delta_l-c_lV(x(t))\leq \delta_l
    \end{equation}
    where $\delta_l < 0$. %and $g_1=c_lV(x(t))\geq0$ and results in
    %\begin{equation}
    %    \dot{V}(x(t))\leq g_2, \quad g_2< 0
    %\end{equation}
%where $g_2 = \delta_l-g_1$. 
Applying Gromwall's lemma, the differential inequality above implies $V(x(t))\leq V(x_0)+\delta_l t $. By definition, the robot is at the exit face when $V(x(t))=0$.% as the $V(x(t))$ shows the distance from the exit face, so 
Hence, $t_{exit}\leq -\frac{d_{max}}{\delta_l}$ and the controller reaches the exit face in finite time if $-\frac{d_{max}}{\delta_l}$ has a finite value.
 \end{proof}

%\rtron{Can we have another proposition to justify the margin on the CBF?}

\begin{proposition}\label{delta_b}
  For each cell, define the set $\cX_{\delta_b}=\{x\in\cX: h(x)\geq -\delta_b\}\subset \cX$. 
  If the problem \eqref{opt-dual} or \eqref{opt-min} is feasible with an optimal $\delta_b$ less than zero and the system relative degree order equals to one, then $\cX_{\delta_b}$ is forward-invariant, and if the robot starts at a distance at least $\delta_b$ from the walls, then it will never get closer than $\delta_b$.
\end{proposition}
\begin{proof}
  The proof follows the original CBF proof. The constraint ensures $\dot{h}+c_bh\geq \delta_b$. Let $y(t)$ be the solution of $\dot{y}+c_by=\delta_b$ with $y(0)=h(x(0))$. $y(0)\geq \delta_b$ by assumption that $x(0)\in\cX_{\delta_b}$. The explicit expression for $y(t)$ is $y(t)=\frac{\delta_b}{c_b}+ce^{-c_bt}$ where $c$ is a constant. Then by the comparison lemma \cite{khalil2002nonlinear}, $h(x(t))\geq y(t)\geq \delta_b$. Hence $x(t)\in \cX_{\delta_b}$ for all $t\geq 0$ and it is forward invariant.
\end{proof}

%%%%%%%%%%%%%%%%%%%%%%%%%%%%%%%%%%%%%%%%%%%%%%%%%%%
\subsection{Stationary Point}\label{sec:stationary}
 This section demonstrates the stability of our controller derived from either \eqref{opt-dual} or \eqref{opt-min}, and establishes that $x_e$ is the equilibrium. Consider the stabilization objective \ref{it:point-stabilization} defined in Section \ref{planning}, and let $x_e$ be the stabilization point in $\cX$. In this section, we provide sufficient conditions showing that the controllers synthesized with our proposed method introduce an equilibrium point at $x_e$.

The stabilization objective \ref{it:point-stabilization} is divided into two cases: first, when the stabilization point is located at the corner of the cell, and second, when the stabilization point is located at the middle of the cell. The distinction is due to the fact that the two cases rely on very different theoretical tools.
%-------------------------------------------------------------------------------------

\subsubsection{Stabilization to a Vertex}
Before proceeding, we need the following. We use $\stack()$ to denote the operator that stacks vertically all its matrix arguments.
\begin{fact}\label{fact:h-cone}
  Let $A_{h,exit}$ be the matrix whose rows are the row vectors in the set $\{A_{h,i}: h_i(x_e)=0\}$. Then $z$ belongs to the proper cone $\{v: A_{h,exit}v\geq 0\}$.
\end{fact}
This fact is intuitively given Definition~\ref{exit_dir}: $A_{h,\textrm{exit}}$ represents the normal of the active constraints at the stabilization point, and $z$ needs to be inward-pointing. Note that the rows or $A_{h,\textrm{exit}}$ are a subset of the rows of $A_{x}$. We can now state the main result of this section.

\begin{proposition}\label{goalPoint}
  Assume the pair $(A,\stack(A_{h,\textrm{exit}},z\transpose))$ is observable and that all $h_i$ and $V$ have the same relative degree $r$ and $x_e$ is one of the corners of the cell $\cX$. Then, any solution to the optimization problem \eqref{opt-feasibility} (or, equivalently, the linear program \eqref{opt-min}) guarantees that $\dot{x}=0$ when $x=x_e$.
\end{proposition}

Note that the assumption about having a homogeneous relative degree is reasonable since $z\transpose$ and $A_{h,\textrm{exit}}$ all essentially represent generic planes in the environment.

\begin{proof}
We prove this by induction. Suppose that $A^i x_e = 0$ for all $i\leq j-1$. We want to prove that $A^{j}x_e = 0$.  As discussed above, we have $V(x_e)=0$ for the Lyapunov function, and $h_i(x_e)=0$ for the constraints corresponding to $A_{h,\textrm{exit}}$. We have that $\psi^v_j(x) = \cL_A^j V + \sum_{i =1}^{j-1} \alpha_i \cL_A ^i V + \alpha_0 V= z\transpose A^j x + \sum_{i =1}^{j-1} \alpha_i A^i x +  \alpha_0 V \leq 0$. Using the induction hypothesis that $A^i x_e = 0$ for all $i \leq j-1$, we simplify $\psi^v_j(x_e)\leq 0$ to $\ z\transpose A^j x_e \leq 0$. A similar argument holds for CBF, and we conclude that $A_{h,\textrm{exit}}A^j x_e \geq 0 $. From Fact~\ref{fact:h-cone}, we have that the sets described by $A_{h,\textrm{exit}}v\geq 0$ and $z\transpose v\leq 0$ intersect only at the point $v=0$; hence, $A^j x_e=0$.Therefore we can conclude that $z\transpose A^{r-1}\dot{x}(xe) = z\transpose A^{r-1}(Ax_e+Bu(x_e)) \leq 0$ and $A_h\transpose A^{r-1}\dot{x}(x_e) = A_h \transpose (Ax_e +Bu(x_e)\geq 0$. Thus, based on the Fact\ref{fact:h-cone}, $\dot{x}(x_e)$ needs to be zero. 

  % Any feasible controller must satisfy the CBF and CLF constraints in~\eqref{opt-feasibility}. As discussed above, we have $V(x_e)=0$ for the Lyapunov function, and $h_i(x_e)=0$ for the constraints corresponding to $A_{h,\textrm{exit}}$. Recalling that $\cL_A^0V=V$, and using the fact the CLF constraint in \eqref{opt-feasibility} %implies Proposition~\ref{prop:ECLF}, claim~\ref{it:der}, 
  % , we have that $\cL_A^j V=z\transpose A^j\dot{x}\leq -c_v \cL_A^{j-1} V= z\transpose A^{j-1}\dot{x}$ for all $0\leq j \leq r$. A similar argument with the CBF constraint in \eqref{opt-feasibility} and %Proposition~\ref{prop:ECBF} 
  % implies that $A_{h,\textrm{exit}}A^j\dot{x}\geq -c_hA_{h,\textrm{exit}}A^{j-1}\dot{x}$.
  % From Fact~\ref{fact:h-cone}, we have that the sets described by $A_{h,\textrm{exit}}v\geq 0$ and $z\transpose v\leq 0$ intersect only at the point $v=0$; hence, $\stack(A_{h,\textrm{exit}},z\transpose) A^j\dot{x}=0$ for all $0\leq j \leq r-1$, which can be compactly described as $\cO_{A} \dot{x}=0$, where $\cO_{A}$ is the observability matrix from the pair $(A,\stack(A_{h,\textrm{exit}},z\transpose))$. Since the latter is observable, $\cO_{A}$ is full rank, and hence $\dot{x}=0$ as claimed.
\end{proof}
Intuitively, the proof shows that the CLF and CBF constraints fix $x$ to $x_e$.
%-----------------------------------------------------------------------------------
\subsubsection{Stabilization to an Inner Point}
In this section, we provide conditions on the feedback control matrix $K$ that are sufficient to imply asymptotic convergence of a given point $x_{eq}$. We first state our results for a generic linear system with closed-loop dynamics
\begin{equation}
  \label{eq:dynamics-simplified}
  \dot{x}=A_Kx,
\end{equation}
and then apply the general result to our case defined by the dynamics~\eqref{sys1} with the output feedback control~\eqref{vec-landmarks},~\eqref{u}.
Without loss of generality, we assume that the equilibrium is at the origin, $x_{eq}=0$ (if not, the same discussion holds after a translation of the coordinate system). We claim the following:
\begin{proposition}\label{prop:inward}
  If there exist a bounded polyhedron $\cX$ such that $0\in\cX$ and the field $\dot{x}=A_Kx$ is \emph{inward-pointing} on $\cX$ (defined as a polytope as in~\eqref{state_limits}), then $x_{eq}=0$ is an asymptotically stable equilibrium of the system $\dot{x}=A_Kx$.
\end{proposition}
The important point of this section is that these convergence conditions on $A_K$ are \emph{linear}, and hence can be easily incorporated into an LP or QP; this is in contrast to standard criteria for stability such as Lyapunov-based conditions ($PA_K+A_K\transpose P = -Q$ for positive definite matrices $P,Q$, which result in a Semi-Definite Programming problem) or algebraic conditions ($A_K$ is Hurwitz \cite{asner1970total}). For completeness, we formally define the notion of \emph{inward-pointing} as follows.
\begin{definition}
  Let $A_{xi},b_{xi}$ denote the $i$-th row of $A_x$ and the $i$-th element of $b_x$, respectively. The vector field $\dot{x}=A_Kx$ is said to be inward-pointing at a point $x_0\in\partial\cX$ if 
  \begin{equation}\label{eq:inward point constraint}
      A_{xi}A_Kx_0\leq \delta_l
  \end{equation}
  for all $i$ such that $A_{xi}x_0=b_i$ and with $\delta_l<0$ strictly negative.
\end{definition}
\begin{definition}\label{def:inward set}
  The vector field $\dot{x}=A_Kx$ is inward-pointing on $\cX$ if it is inward-pointing for every $x\in\partial\cX$ with a common $\delta_l$.
\end{definition}
Note that, by Nagumo’s theorem \cite{nagumo1942lage}, Definition~\ref{def:inward set} implies that $\cX$ is forward invariant.

In order to prove Proposition~\ref{prop:inward}, we first need to review the following fact from linear dynamical system theory:
\begin{fact}\label{fact:linear systems} Assume that $x_{eq}=0$ is an equilibrium of the linear dynamical system $\dot{x}=A_Kx$. Then, the system will exhibit one of the following behaviors:
\begin{enumerate}
\item\label{it:gas} Globally asymptotically stable: all eigenvalues have a negative real part.
\item The system is stable, and all trajectories converge to a linear subspace $\cS_{x_{eq}}$ containing $x_{eq}$: some eigenvalues are zero ($A_K$ is singular), and the other have a negative real part.
\item The system converges to a bounded periodic orbit $\cO$ that does not contain $x_{eq}$: some eigenvalues are complex conjugate with a zero real part, and the others have a negative real part.
\item The system is unstable: at least one eigenvalue has a positive real part.
\end{enumerate}
\end{fact}
This fact can be easily proven by looking at the closed-form solution for $x(t)$ (for any arbitrary initial condition $x(0)\in\real{n}$) via the classical variation of the constants formula.
Second, the geometry of our setup implies the following:
\begin{lemma}\label{lemma:scaling}
  Let $\cX_s$ be a scaled version of $\cX$ defined as $\cX_s=\{x\mid A_xx\leq sb_x\}$ for a given scale $s>0$. Then, the field $\dot{x}=A_Kx$ is inward-pointing on $\cX$ if and only if it will also be inward-pointing on $\cX_s$.
\end{lemma}
\begin{proof}
  The claim follows by applying definition~\ref{def:inward set} with $s\delta_l$.
\end{proof}
Finally, we have all the elements necessary for proving the main claim of this section.

\begin{proof}[Proposition~\ref{prop:inward}]
  We proceed by contradiction to exclude all possible behaviors listed in Fact~\ref{fact:linear systems} except~\ref{it:gas} global asymptotic convergence.
  \begin{enumerate}
    \setcounter{enumi}{1}
  \item Let $x_0$ be a point in the intersection $\cS_{x_{eq}}\cap \partial\cX$. Then $A_Kx_0=0$, and the point violates the inner-pointing assumption.
  \item Let $s_{1},s_{2}>0$ be two scales such that $\cX_{s_{1}}$ does not contain $\cO$, $\cX_{s_{1}}\cap \cO=\emptyset$, and $\cX_{s_{2}}$ contains $\cO$, $\cO\subset\cX_{s_{2}}$. Then, there exists an $s$ such that $s_1<s<s_2$ and $\cX_s$ intersects $\cO$, $\cX_s\cap\cO=x_0$ for some $x_0\in\cO$ (this follows from the assumption that $\cX_s$ is compact, and by applying the intermediate value theorem to the signed distance between $x_0$ and the set $\cX_s$ as a function of $s$). Then, the periodic solution escapes $\cX_s$ at $x_0$, and thus cannot be inner pointing at that point; by Lemma~\ref{lemma:scaling}, this creates a contradiction with the assumption that the field is inner pointing on $\cX$.
  \item Since the system is unstable, at least one trajectory must escapes $\cX$ at a point $x_0$; again, this contradicts the assumption that the field is inner pointing on $\cX$.
  \end{enumerate}
\end{proof}
\begin{definition}
  We define a \textit{non-final} cell as a cell the agent passes through to reach the goal point. The \textit{final} cell is defined as a cell in which the goal point belongs to that cell. 
\end{definition}
\begin{theorem}
If $x_e$ is in the interior of a cell, and if the solution of (final opt problem) implies $\delta_b<0$, then the controllers $K_i$ will stabilize the system to $x_e$.
\end{theorem}
\begin{proof}
The non-final cells do not have an equilibrium due to Proposition 1.
For the final cell, the CBF condition \eqref{primal-cbf}  evaluated after substituting our controller, and at the boundary of the cell, is equivalent to \eqref{eq:inward point constraint}. The CBF dual constraint \eqref{dual-conscbf} then implies that \eqref{eq:inward point constraint} is satisfied for every point in the final cell. Hence, $\dot{x}$ produced from our controller satisfies Definition \ref{def:inward set} in the final cell.
The claim is then a consequence of Proposition \ref{prop:inward}.
\end{proof}

%%%%%%%%%%%%%%%%%%%%%%%%%%%%%%%%%%%%%%%%%%%%%%%%%%%%%%%%%%%%%%%%%%%%%%%%%%%%%%%%%%%%%%%%%%%%%%%%%%%
\subsection{Control With the Limited Field of View}\label{sec:limited-field-of-view}
% One of the drawbacks of driving an agent in a real-world environment is that i
In the formulation above, it is implicitly assumed that the controller has access to all the landmarks measurements at all times. However, in practice, a robot will only be able to detect a subset of the landmarks due to a limited field of view or environment occlusions. To tackle this issue, we show in this section that the controller $u$ \eqref{u} can be designed using multiple landmarks (as in the preceding section) but then computed using a subset of landmarks.
\begin{proposition} \label{prop:limited field view}
 % \rtron{Update with $k$}
  Let $K_p=\bmat{K_{P1},\cdots,K_{Pi},\cdots, K_{Pl}}$ be a partition of the controller matrix conformal with $L^\vee$. Without loss of generality, assume that we see all the landmarks with $i\leq\hat{\imath}$, while landmarks $\hat{\imath}+1\leq i \leq l$ are not visible. Then the controller \eqref{u} can be equivalently written as
  \begin{equation}\label{u_new}
    u = \sum_{i=1}^{\hat{\imath}} K_{Pi} y_i +\sum_{i=\hat{\imath}+1}^l K_{Pj}y_{\hat{\imath}}+k_{\text{bias}}+K_b+ K_dx_d,
  \end{equation}
  where $k_{\text{bias},i}\in\real{n}$ is a constant vector given by
  \begin{equation}
    k_{\text{bias},i}=\sum_{j=
    \hat{\imath}+1}^{l}K_{Pj}(\hat{l}_j-\hat{l}_{\hat{\imath}})
  \end{equation}
\end{proposition}
\begin{proof}
  Using the conformal partition of $K_p$, we can expand \eqref{u} as
  \begin{equation}\label{decomU}
    u_p = \sum_jK_{pj}(\hat{l}_j-x).
  \end{equation}
  Adding and subtracting $\sum_{j=
    \hat{\imath}+1}^{l} K_{pj} y_{\hat{\imath}}=\sum_j K_{pj}(\hat{l}_{\hat{\imath}}-x)$ (note that $K_p$ and $y$ have different subscripts) and reordering, we have
  \begin{equation}\label{u_limited}
      u_p = \sum_{i=1}^{\hat{\imath}} K_{pi} y_i +
       \sum_{j=\hat{\imath}+1}^l K_{pj}(\hat{l}_{\hat{\imath}}-x)+
       \sum_{j=
    \hat{\imath}+1}^{l}K_{pj}(\hat{l}_j-\hat{l}_{\hat{\imath}}),
  \end{equation}
  from which the claim follows.
\end{proof}

Note that we could also merge all visible landmarks into a single virtual landmark, which can be chosen to minimize (at every time instant) the effect of noise \cite{wang2021chance}.

Using the fact that the global positions of the landmarks are known during planning, our new Proposition~\ref{prop:limited field view} shows that it is possible to implement the controller $u$ by measuring a single displacement $y_i$; moreover, since the original controller \eqref{u} is smooth, one can also switch among different landmarks without introducing discontinuities in the control. Although we stated our result for a single landmark, it is possible to prove a similar claim for any subset of landmarks.

%%%%%%%%%%%%%%%%%%%%%%%%%%%%%%%%%%%%%%%%%%%%%%%%%%%%%%%%%%%%%%%%%%%%%%%%%%%%%%%%%%%%%%%%%%%%%%%%%%%
\section{Control With Bearing Measurements}\label{sec:bearing-measurments}
As mentioned in the introduction, in this paper, we consider a robot equipped with a monocular camera that generally does not provide the depth of a target object in the image and instead measures the corresponding relative bearing (Sec.~\ref{sec:bearing}).
In this section, we only consider a driftless system(i.e., $A=0$) with a relative degree of one and  show that in this case, it is possible to still use the control synthesis method of Sec.~\ref{sec:problem setup} after rescaling the bearing measurements, such that they are similar to the ideal displacement measurements. We divide the section into two parts. First, we give details on the rescaling procedure; second, we show that the resulting bearing controller still solves the path planning problem, albeit with modified CLF and CBF conditions.
%%%%%%%%%%%%%%%%%%%%%%%%%%%%%%%%%%%%%%%%%%%%%%%%%%%%%%%%%%
\subsection{Bearing direction measurements}\label{sec:bearing}
In this work, we assume global compass direction is available (the bearing directions can
be compared in the same frame of reference) and the robot has only access to bearing direction measurements (see Fig.~\ref{bearing}). The bearing direction measurements are defined as
\begin{equation}\label{bearing_mes}
    \beta_i= d_i(x)^{-1}(l_i-x),\;\;\;\; i=\{1,\hdots,N\},
\end{equation}
where $N$ is the number of \emph{landmarks} and $d_i$ is the relative displacement measurement between the position of the robot and landmark $l_i$
\begin{equation}\label{distance}
    d_i(x_p) = \norm{l_i-x_p}.
\end{equation}
The location of landmarks ${l}_i$ is assumed to be fixed and known at the planning time. Additionally, we assume at the implementation time, the robot can measure the bearing direction $\beta_i$ between its position $x_p$ and a set of landmarks $l_i$.

\begin{figure*}
\subfloat[]{\label{bearing}\scalebox{.75}{\tikzset{every picture/.style={line width=0.75pt}} %set default line width to 0.75pt        

\begin{tikzpicture}[x=0.75pt,y=0.75pt,yscale=-1,xscale=1]

%Shape: Ellipse [id:dp6671379829837052] 
\draw  [color={rgb, 255:red, 19; green, 27; blue, 218 }  ,draw opacity=1 ][fill={rgb, 255:red, 28; green, 23; blue, 213 }  ,fill opacity=1 ] (74.52,94.22) .. controls (74.52,89.58) and (78.19,85.81) .. (82.71,85.81) .. controls (87.24,85.81) and (90.9,89.58) .. (90.9,94.22) .. controls (90.9,98.87) and (87.24,102.63) .. (82.71,102.63) .. controls (78.19,102.63) and (74.52,98.87) .. (74.52,94.22) -- cycle ;
%Shape: Ellipse [id:dp009466765890195461] 
\draw  [color={rgb, 255:red, 19; green, 27; blue, 218 }  ,draw opacity=1 ][fill={rgb, 255:red, 28; green, 23; blue, 213 }  ,fill opacity=1 ] (176.6,58.23) .. controls (176.6,53.58) and (180.27,49.82) .. (184.79,49.82) .. controls (189.31,49.82) and (192.98,53.58) .. (192.98,58.23) .. controls (192.98,62.87) and (189.31,66.64) .. (184.79,66.64) .. controls (180.27,66.64) and (176.6,62.87) .. (176.6,58.23) -- cycle ;
%Shape: Ellipse [id:dp9773660276649394] 
\draw  [color={rgb, 255:red, 19; green, 27; blue, 218 }  ,draw opacity=1 ][fill={rgb, 255:red, 28; green, 23; blue, 213 }  ,fill opacity=1 ] (67.67,204.55) .. controls (67.67,199.91) and (71.33,196.14) .. (75.86,196.14) .. controls (80.38,196.14) and (84.05,199.91) .. (84.05,204.55) .. controls (84.05,209.2) and (80.38,212.96) .. (75.86,212.96) .. controls (71.33,212.96) and (67.67,209.2) .. (67.67,204.55) -- cycle ;
%Straight Lines [id:da7430980107932932] 
\draw [color={rgb, 255:red, 74; green, 144; blue, 226 }  ,draw opacity=1 ] [dash pattern={on 0.84pt off 2.51pt}]  (296.77,184.99) -- (192.7,68.87) ;
\draw [shift={(190.69,66.64)}, rotate = 48.13] [fill={rgb, 255:red, 74; green, 144; blue, 226 }  ,fill opacity=1 ][line width=0.08]  [draw opacity=0] (8.93,-4.29) -- (0,0) -- (8.93,4.29) -- cycle    ;
%Straight Lines [id:da9071661834087201] 
\draw [color={rgb, 255:red, 74; green, 144; blue, 226 }  ,draw opacity=1 ] [dash pattern={on 0.84pt off 2.51pt}]  (296.77,184.99) -- (90.62,99.09) ;
\draw [shift={(87.86,97.94)}, rotate = 22.62] [fill={rgb, 255:red, 74; green, 144; blue, 226 }  ,fill opacity=1 ][line width=0.08]  [draw opacity=0] (8.93,-4.29) -- (0,0) -- (8.93,4.29) -- cycle    ;
%Straight Lines [id:da37734478892531786] 
\draw [color={rgb, 255:red, 74; green, 144; blue, 226 }  ,draw opacity=1 ] [dash pattern={on 0.84pt off 2.51pt}]  (296.77,184.99) -- (87.03,204.28) ;
\draw [shift={(84.05,204.55)}, rotate = 354.75] [fill={rgb, 255:red, 74; green, 144; blue, 226 }  ,fill opacity=1 ][line width=0.08]  [draw opacity=0] (8.93,-4.29) -- (0,0) -- (8.93,4.29) -- cycle    ;
%Straight Lines [id:da49085022616755736] 
\draw [color={rgb, 255:red, 126; green, 211; blue, 33 }  ,draw opacity=1 ] [dash pattern={on 0.84pt off 2.51pt}]  (79.38,97.06) -- (177.46,61.65) ;
%Straight Lines [id:da7342706435337161] 
\draw [color={rgb, 255:red, 126; green, 211; blue, 33 }  ,draw opacity=1 ] [dash pattern={on 0.84pt off 2.51pt}]  (75.86,204.55) -- (82.71,94.22) ;
%Straight Lines [id:da360270235780652] 
\draw [color={rgb, 255:red, 126; green, 211; blue, 33 }  ,draw opacity=1 ] [dash pattern={on 0.84pt off 2.51pt}]  (75.86,204.55) -- (184.79,58.23) ;
%Straight Lines [id:da9519085103220144] 
\draw [color={rgb, 255:red, 245; green, 166; blue, 35 }  ,draw opacity=1 ]   (75.86,204.55) -- (103.55,168.4) ;
\draw [shift={(105.38,166.01)}, rotate = 127.45] [fill={rgb, 255:red, 245; green, 166; blue, 35 }  ,fill opacity=1 ][line width=0.08]  [draw opacity=0] (10.72,-5.15) -- (0,0) -- (10.72,5.15) -- (7.12,0) -- cycle    ;
%Straight Lines [id:da025314403613906133] 
\draw [color={rgb, 255:red, 245; green, 166; blue, 35 }  ,draw opacity=1 ]   (75.86,204.55) -- (77.11,158.06) ;
\draw [shift={(77.19,155.06)}, rotate = 91.54] [fill={rgb, 255:red, 245; green, 166; blue, 35 }  ,fill opacity=1 ][line width=0.08]  [draw opacity=0] (10.72,-5.15) -- (0,0) -- (10.72,5.15) -- (7.12,0) -- cycle    ;
%Straight Lines [id:da2008860912113659] 
\draw [color={rgb, 255:red, 208; green, 2; blue, 27 }  ,draw opacity=1 ]   (82.71,94.22) -- (79.69,135.64) ;
\draw [shift={(79.48,138.63)}, rotate = 274.17] [fill={rgb, 255:red, 208; green, 2; blue, 27 }  ,fill opacity=1 ][line width=0.08]  [draw opacity=0] (10.72,-5.15) -- (0,0) -- (10.72,5.15) -- (7.12,0) -- cycle    ;
%Straight Lines [id:da7543949209873959] 
\draw [color={rgb, 255:red, 208; green, 2; blue, 27 }  ,draw opacity=1 ]   (82.71,94.22) -- (125.57,80.28) ;
\draw [shift={(128.42,79.35)}, rotate = 161.98] [fill={rgb, 255:red, 208; green, 2; blue, 27 }  ,fill opacity=1 ][line width=0.08]  [draw opacity=0] (10.72,-5.15) -- (0,0) -- (10.72,5.15) -- (7.12,0) -- cycle    ;
%Straight Lines [id:da6605041976092756] 
\draw [color={rgb, 255:red, 82; green, 135; blue, 20 }  ,draw opacity=1 ]   (184.79,58.23) -- (144.71,74.88) ;
\draw [shift={(141.94,76.03)}, rotate = 337.44] [fill={rgb, 255:red, 82; green, 135; blue, 20 }  ,fill opacity=1 ][line width=0.08]  [draw opacity=0] (10.72,-5.15) -- (0,0) -- (10.72,5.15) -- (7.12,0) -- cycle    ;
%Straight Lines [id:da7643051466755899] 
\draw [color={rgb, 255:red, 82; green, 135; blue, 20 }  ,draw opacity=1 ]   (184.79,58.23) -- (159.67,93.92) ;
\draw [shift={(157.94,96.37)}, rotate = 305.14] [fill={rgb, 255:red, 82; green, 135; blue, 20 }  ,fill opacity=1 ][line width=0.08]  [draw opacity=0] (10.72,-5.15) -- (0,0) -- (10.72,5.15) -- (7.12,0) -- cycle    ;
%Shape: Ellipse [id:dp8168430375255511] 
\draw  [color={rgb, 255:red, 0; green, 0; blue, 0 }  ,draw opacity=1 ][fill={rgb, 255:red, 0; green, 0; blue, 0 }  ,fill opacity=1 ] (288.58,184.99) .. controls (288.58,180.34) and (292.25,176.58) .. (296.77,176.58) .. controls (301.29,176.58) and (304.96,180.34) .. (304.96,184.99) .. controls (304.96,189.64) and (301.29,193.4) .. (296.77,193.4) .. controls (292.25,193.4) and (288.58,189.64) .. (288.58,184.99) -- cycle ;
%Straight Lines [id:da43962304847760536] 
\draw [color={rgb, 255:red, 74; green, 144; blue, 226 }  ,draw opacity=1 ]   (296.77,184.99) -- (263.84,148.49) ;
\draw [shift={(262.5,147)}, rotate = 47.95] [fill={rgb, 255:red, 74; green, 144; blue, 226 }  ,fill opacity=1 ][line width=0.08]  [draw opacity=0] (12,-3) -- (0,0) -- (12,3) -- cycle    ;
%Straight Lines [id:da9959943770434769] 
\draw [color={rgb, 255:red, 74; green, 144; blue, 226 }  ,draw opacity=1 ]   (296.77,184.99) -- (253.44,166.03) ;
\draw [shift={(251.61,165.23)}, rotate = 23.63] [fill={rgb, 255:red, 74; green, 144; blue, 226 }  ,fill opacity=1 ][line width=0.08]  [draw opacity=0] (12,-3) -- (0,0) -- (12,3) -- cycle    ;
%Straight Lines [id:da6585426247212725] 
\draw [color={rgb, 255:red, 74; green, 144; blue, 226 }  ,draw opacity=1 ]   (296.77,184.99) -- (248.49,189.96) ;
\draw [shift={(246.5,190.17)}, rotate = 354.12] [fill={rgb, 255:red, 74; green, 144; blue, 226 }  ,fill opacity=1 ][line width=0.08]  [draw opacity=0] (12,-3) -- (0,0) -- (12,3) -- cycle    ;

% Text Node
\draw (305.77,178.05) node [anchor=north west][inner sep=0.75pt]    {$x$};
% Text Node
\draw (202.36,34.17) node [anchor=north west][inner sep=0.75pt]    {$l_{1}$};
% Text Node
\draw (65.24,54.52) node [anchor=north west][inner sep=0.75pt]    {$l_{2}$};
% Text Node
\draw (48.48,182.06) node [anchor=north west][inner sep=0.75pt]    {$l_{3}$};
% Text Node
\draw (92.52,63.52) node [anchor=north west][inner sep=0.75pt]  [font=\tiny]  {$\beta _{2}^{1}$};
% Text Node
\draw (58.24,106.55) node [anchor=north west][inner sep=0.75pt]  [font=\tiny]  {$\beta _{2}^{3}$};
% Text Node
\draw (111.56,169.94) node [anchor=north west][inner sep=0.75pt]  [font=\tiny]  {$\beta _{3}^{1}$};
% Text Node
\draw (58.24,154.29) node [anchor=north west][inner sep=0.75pt]  [font=\tiny]  {$\beta _{3}^{2}$};
% Text Node
\draw (174.79,92.47) node [anchor=north west][inner sep=0.75pt]  [font=\tiny]  {$\beta _{1}^{3}$};
% Text Node
\draw (150.41,45.52) node [anchor=north west][inner sep=0.75pt]  [font=\tiny]  {$\beta _{1}^{2}$};
% Text Node
\draw (262.56,188.94) node [anchor=north west][inner sep=0.75pt]  [font=\scriptsize]  {$\beta _{3}$};
% Text Node
\draw (264.06,172.44) node [anchor=north west][inner sep=0.75pt]  [font=\scriptsize]  {$\beta _{2}$};
% Text Node
\draw (276.56,150.44) node [anchor=north west][inner sep=0.75pt]  [font=\scriptsize]  {$\beta _{1}$};

\end{tikzpicture}}}
\hfill
 \subfloat[]{\label{circle}\scalebox{.75}{\input{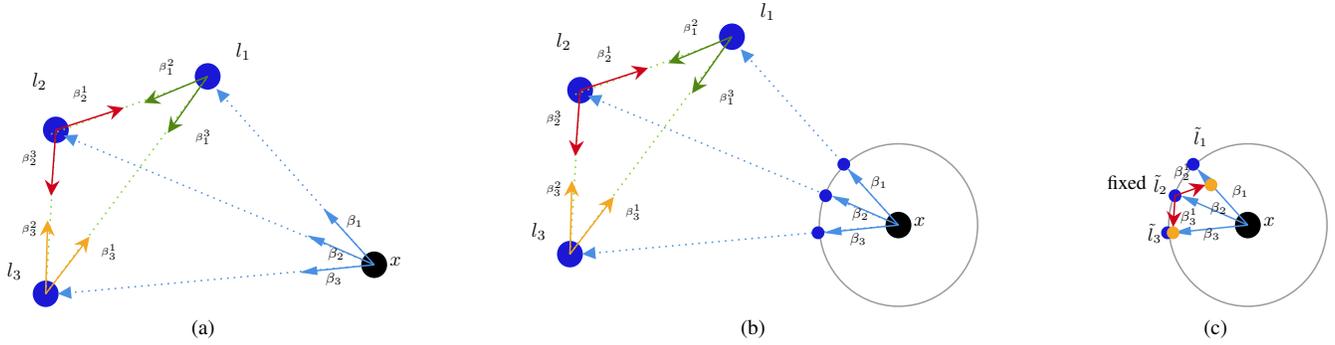}}}
 \hfill
  \subfloat[]{\label{config}\scalebox{.75}{\tikzset{every picture/.style={line width=0.75pt}} %set default line width to 0.75pt        

\begin{tikzpicture}[x=0.75pt,y=0.75pt,yscale=-1,xscale=1]

%Shape: Ellipse [id:dp6208519030063377] 
\draw  [color={rgb, 255:red, 155; green, 155; blue, 155 }  ,draw opacity=1 ] (263.54,204.99) .. controls (263.54,174.79) and (287.37,150.31) .. (316.77,150.31) .. controls (346.17,150.31) and (370,174.79) .. (370,204.99) .. controls (370,235.19) and (346.17,259.67) .. (316.77,259.67) .. controls (287.37,259.67) and (263.54,235.19) .. (263.54,204.99) -- cycle ;
%Shape: Ellipse [id:dp23329950411406952] 
\draw  [color={rgb, 255:red, 19; green, 27; blue, 218 }  ,draw opacity=1 ][fill={rgb, 255:red, 28; green, 23; blue, 213 }  ,fill opacity=1 ] (264.04,185.23) .. controls (264.04,183.12) and (265.74,181.42) .. (267.83,181.42) .. controls (269.92,181.42) and (271.61,183.12) .. (271.61,185.23) .. controls (271.61,187.34) and (269.92,189.05) .. (267.83,189.05) .. controls (265.74,189.05) and (264.04,187.34) .. (264.04,185.23) -- cycle ;
%Shape: Ellipse [id:dp16805462178397468] 
\draw  [color={rgb, 255:red, 19; green, 27; blue, 218 }  ,draw opacity=1 ][fill={rgb, 255:red, 28; green, 23; blue, 213 }  ,fill opacity=1 ] (276.23,164.11) .. controls (276.23,162) and (277.93,160.29) .. (280.02,160.29) .. controls (282.11,160.29) and (283.8,162) .. (283.8,164.11) .. controls (283.8,166.21) and (282.11,167.92) .. (280.02,167.92) .. controls (277.93,167.92) and (276.23,166.21) .. (276.23,164.11) -- cycle ;
%Shape: Ellipse [id:dp20028604752200274] 
\draw  [color={rgb, 255:red, 19; green, 27; blue, 218 }  ,draw opacity=1 ][fill={rgb, 255:red, 28; green, 23; blue, 213 }  ,fill opacity=1 ] (258.93,210.17) .. controls (258.93,208.06) and (260.63,206.35) .. (262.72,206.35) .. controls (264.81,206.35) and (266.5,208.06) .. (266.5,210.17) .. controls (266.5,212.27) and (264.81,213.98) .. (262.72,213.98) .. controls (260.63,213.98) and (258.93,212.27) .. (258.93,210.17) -- cycle ;
%Shape: Ellipse [id:dp9924691795934781] 
\draw  [color={rgb, 255:red, 0; green, 0; blue, 0 }  ,draw opacity=1 ][fill={rgb, 255:red, 0; green, 0; blue, 0 }  ,fill opacity=1 ] (308.58,204.99) .. controls (308.58,200.34) and (312.25,196.58) .. (316.77,196.58) .. controls (321.29,196.58) and (324.96,200.34) .. (324.96,204.99) .. controls (324.96,209.64) and (321.29,213.4) .. (316.77,213.4) .. controls (312.25,213.4) and (308.58,209.64) .. (308.58,204.99) -- cycle ;
%Straight Lines [id:da21139704961749506] 
\draw [color={rgb, 255:red, 74; green, 144; blue, 226 }  ,draw opacity=1 ]   (316.77,204.99) -- (283.84,168.49) ;
\draw [shift={(282.5,167)}, rotate = 47.95] [fill={rgb, 255:red, 74; green, 144; blue, 226 }  ,fill opacity=1 ][line width=0.08]  [draw opacity=0] (12,-3) -- (0,0) -- (12,3) -- cycle    ;
%Straight Lines [id:da7742397851408793] 
\draw [color={rgb, 255:red, 74; green, 144; blue, 226 }  ,draw opacity=1 ]   (316.77,204.99) -- (273.44,186.03) ;
\draw [shift={(271.61,185.23)}, rotate = 23.63] [fill={rgb, 255:red, 74; green, 144; blue, 226 }  ,fill opacity=1 ][line width=0.08]  [draw opacity=0] (12,-3) -- (0,0) -- (12,3) -- cycle    ;
%Straight Lines [id:da9618637590142907] 
\draw [color={rgb, 255:red, 74; green, 144; blue, 226 }  ,draw opacity=1 ]   (316.77,204.99) -- (268.49,209.96) ;
\draw [shift={(266.5,210.17)}, rotate = 354.12] [fill={rgb, 255:red, 74; green, 144; blue, 226 }  ,fill opacity=1 ][line width=0.08]  [draw opacity=0] (12,-3) -- (0,0) -- (12,3) -- cycle    ;
%Shape: Ellipse [id:dp7804457811630825] 
\draw  [color={rgb, 255:red, 245; green, 166; blue, 35 }  ,draw opacity=1 ][fill={rgb, 255:red, 245; green, 166; blue, 35 }  ,fill opacity=1 ] (288.22,178) .. controls (288.22,175.89) and (289.91,174.18) .. (292,174.18) .. controls (294.09,174.18) and (295.78,175.89) .. (295.78,178) .. controls (295.78,180.11) and (294.09,181.82) .. (292,181.82) .. controls (289.91,181.82) and (288.22,180.11) .. (288.22,178) -- cycle ;
%Straight Lines [id:da32681756593815847] 
\draw [color={rgb, 255:red, 208; green, 2; blue, 27 }  ,draw opacity=1 ]   (267.83,185.23) -- (285.39,179) ;
\draw [shift={(288.22,178)}, rotate = 160.47] [fill={rgb, 255:red, 208; green, 2; blue, 27 }  ,fill opacity=1 ][line width=0.08]  [draw opacity=0] (10.72,-5.15) -- (0,0) -- (10.72,5.15) -- (7.12,0) -- cycle    ;
%Shape: Ellipse [id:dp7538722918153953] 
\draw  [color={rgb, 255:red, 245; green, 166; blue, 35 }  ,draw opacity=1 ][fill={rgb, 255:red, 245; green, 166; blue, 35 }  ,fill opacity=1 ] (262.72,210.17) .. controls (262.72,208.06) and (264.41,206.35) .. (266.5,206.35) .. controls (268.59,206.35) and (270.28,208.06) .. (270.28,210.17) .. controls (270.28,212.27) and (268.59,213.98) .. (266.5,213.98) .. controls (264.41,213.98) and (262.72,212.27) .. (262.72,210.17) -- cycle ;
%Straight Lines [id:da8989091929483379] 
\draw [color={rgb, 255:red, 208; green, 2; blue, 27 }  ,draw opacity=1 ]   (267.83,185.23) -- (266.69,203.36) ;
\draw [shift={(266.5,206.35)}, rotate = 273.6] [fill={rgb, 255:red, 208; green, 2; blue, 27 }  ,fill opacity=1 ][line width=0.08]  [draw opacity=0] (10.72,-5.15) -- (0,0) -- (10.72,5.15) -- (7.12,0) -- cycle    ;

% Text Node
\draw (325.77,198.05) node [anchor=north west][inner sep=0.75pt]    {$x$};
% Text Node
\draw (247.98,202.56) node [anchor=north west][inner sep=0.75pt]  [font=\small]  {$\tilde{l}_{3}$};
% Text Node
\draw (251.98,170.06) node [anchor=north west][inner sep=0.75pt]  [font=\small]  {$\tilde{l}_{2}$};
% Text Node
\draw (278.98,138.06) node [anchor=north west][inner sep=0.75pt]  [font=\small]  {$\tilde{l}_{1}$};
% Text Node
\draw (220.67,170.06) node [anchor=north west][inner sep=0.75pt]   [align=left] {fixed};
% Text Node
\draw (265.33,162.4) node [anchor=north west][inner sep=0.75pt]  [font=\scriptsize]  {$\beta _{2}^{1}$};
% Text Node
\draw (269.83,192.45) node [anchor=north west][inner sep=0.75pt]  [font=\scriptsize]  {$\beta _{3}^{1}$};
% Text Node
\draw (304.67,177.07) node [anchor=north west][inner sep=0.75pt]  [font=\scriptsize]  {$\beta _{1}$};
% Text Node
\draw (285.33,205.73) node [anchor=north west][inner sep=0.75pt]  [font=\scriptsize]  {$\beta _{3}$};
% Text Node
\draw (290,188.4) node [anchor=north west][inner sep=0.75pt]  [font=\scriptsize]  {$\beta _{2}$};

\end{tikzpicture}}}
  \caption{a) Bearing direction measurements. Landmarks are shown by blue circles, and the robot is shown by a black circle. We assume the robot measures the bearing to landmarks directly and they are shown by blue arrows. The position of landmarks is known, and bearing direction measurements between landmarks are shown by purple, green, and orange arrows. b) As bearing measurements are unit vectors, the robot sees all the landmarks within 1 unit from itself. c) modify bearing measurements by assuming the landmark $l_2$ is fixed.}
\end{figure*}
%%%%%%%%%%%%%%%%%%%%%%%%%%%%%%%%%%%%%%%%%%%%%%%%%%%%%%%%%%
For simplicity, in this section, we consider only 2-D environments (however, an extension to the 3-D case is possible with minor modifications). Without loss of generality, we assume that the robot position $x$ is equal to the zero vector (if not, all the computations below are still valid after introducing opportune shifts by $x$).

The bearings $\{\beta_i\}$ can be identified with points on a unitary circle centered at the origin (Fig.~\ref{bearing}).

For the current cell (and hence its associated controller) we pick a \emph{fixed} landmark $f$ among those available. We then define $\tilde{l}_f$ as the point on the circle corresponding to $\beta_f$. Our goal is to rescale all the other bearings $i\neq f$ such that they are similar to the full displacements.

For this purpose, we define \emph{inter-landmark} bearing directions between landmark $f$ and every other landmarks $i\neq f$ as
\begin{equation}
    \beta_i^f= d_i(l_f)^{-1}(l_i-l_f);
\end{equation}
note that these bearings can be pre-computed from the known locations of the landmarks on the map.
For each landmark, $i\neq f$, let $P_i^f$ be the line passing through the fixed landmark $\tilde{l}_f$ with direction $\beta_i^f$, and let $P_i$ be the line passing through the robot position $x$ with direction $\beta_i$. We then define the scaled landmark position $\tilde{l}_i$ as the intersection between $P_i^f$ and $P_i$.

\begin{lemma}
Let $s(x_p)=d_f(x_P)$ be the distance between landmark $f$ and the robot; then, we have that $l_i-x=s(x_p) (\tilde{l}_i-x_p)$ for each landmark $i$.
\end{lemma}
\begin{proof}
The triangles $l_f,x_p,l_i$ and $\tilde{l}_f,x_p,\tilde{l}_i$ are similar since they have identical internal angles. Moreover, $\norm{\tilde{l}_f-x_p}=1$ by construction, the ratio between the segments $l_f,x_p$ and $\tilde{l}_f,x_p$ is equal to $s(x_p)$. Combining these two facts, we have that the ratio between the segments $l_f,x_p$ and $\tilde{l}_f,x_p$ is also $s(x_p)$; the claim then follows.
\end{proof}

Our proposed solution is then to compute the scale displacements $\tilde{\cY}=\stack(\{\tilde{l}_i-x_p\})$, which are then used with the pre-computed controller
\begin{equation}\label{eq:u tilde}
\tilde{u}_{ij}(x)=K_{ij}\tilde{\cY} .
\end{equation}
\subsection{Analysis of the Bearing Controller}%
\label{sec:bearing_measurement}

The following lemma shows that the original displacement-based controller $u_{ij}$ and our proposed bearing-based controller $\tilde{u}_{ij}$ are essentially equivalent from the point of view of path planning.

\begin{proposition}\label{prop:speed reparametrization}
Assume $s(x_p)$ is uniformly upper bounded (i.e., $s(x_p)<\infty$ for all $x_p\in \cX_{ij}$). The controllers $u_{ij}$ in \eqref{u} and $\tilde{u}_{ij}$ in \eqref{eq:u tilde} produce the same paths (but traced, in general, with different speeds) for the driftless system \eqref{sys1} (i.e., $A=0$) when started from the same initial condition.
\end{proposition}
\begin{proof}
Let $x_p$ and $\tilde{x_p}$ be the trajectories of the system under $u_{ij}$ and $\tilde{u}_{ij}$, respectively. Since both the dynamics and the controllers are linear, we have that $\dot{x_p}=s\dot{\tilde{x_p}}$ when evaluated at the same location. This implies that the two curves $x_p$ and $\tilde{x_p}$ are the same up to a reparametrization of the velocity.
\end{proof}

In fact, we can also relate the new controller to the conditions in the synthesis problem~\eqref{opt-dual}.
\begin{proposition}\label{prop:new CBF CLF conditions}
Assume that $s_{\min}<s(x)\leq s_{\max}$, and that $\tilde{u}=K_{ij}s\cY\in \cU$ for all $x\in \cX_{ij}$. Then $K_{ij}$ is a feasible solution for \eqref{opt-dual} with the modified CLF and CBF conditions:
\begin{align}
-( \cL_Bh_{ij}\tilde{u}+\tilde{c}_hh_{ij})&\leq 0, \\ \cL_BV_{ij}\tilde{u}+\tilde{c}_vV_{ij}&\leq 0,  \end{align}
where $\tilde{c}_h=s_{\min}\inverse c_h$ and $\tilde{c}_v=s_{\max}\inverse c_v$.
\end{proposition}
\begin{proof} The claim follows by dividing the original CBF and CLF conditions by $s(x)$, and then using the bounds $s_{\min}$, $s_{\max}$.
\end{proof}

Note that the bounds on $s(x)$ translate to bounds on the distance between the cell $\cX_{ij}$ and the landmarks $l_i$

Taken together, Propositions~\ref{prop:speed reparametrization} and~\ref{prop:new CBF CLF conditions} show that the controller $K_{ij}$ found by assuming a displacement-based controller can also be used for the bearing-based case. However, the speed of the resulting trajectories might be more aggressive.
%%%%%%%%%%%%%%%%%%%%%%%%%%%%%%%%%%%%%%%%%
\section{Numerical Examples}\label{sec:examples}
To assess the effectiveness of the proposed algorithm, we run a set of validations using MATLAB simulations. While the optimization problem guarantees convergence of the robot to the stabilization point, in these experiments the velocity control input $u$ has been normalized to achieve constant velocities along the robot's trajectory.

The simulated MATLAB environment is presented in Fig.~\ref{fig:O1} and Fig.~\ref{fig:O2}. The Polygonal environment is decomposed into six convex cells. Each cell has four vertices. This experiment tests the effect of the $\eta$ on the smoothness of the controller when switching between cells. This experiment also represents the effect of choosing landmarks. In Fig.~\ref{fig:stationary_share}, all six cells share the same set of landmarks, which is a set of all vertices and are shown by blue markers. In this figure, the robot starts from the start point and passes through other cells to reach the goal point. As represented in Fig.~\ref{fig:stationary_share}, as we increase the $\eta$ the path becomes smoother. The non-regularized refers to the case where the cost function is computed such that the $\phi^t$ and $\phi_p$ are eliminated from the cost function, so we do not consider the smoothness of the path between cells.  In Fig.~\ref{fig:stationary_diff}, upper cells, shown by orange edges, take the measurements from the red landmarks, and the lower cells, shown by black edges, take measurements from the blue landmarks. Comparing the two Fig.~\ref{fig:stationary_share} and Fig.~\ref{fig:stationary_diff} shows that this approach produces a smoother path when all cells share the same set of landmarks. 
In Fig.~\ref{fig:o2-shared} and Fig.~\ref{fig:o2-diff}, the robot moves through the feasible path to cover the environment. Similar to Fig.~\ref{fig:O1}, we separate the cases where the cells get inputs from all landmarks and the cases where different landmarks are assigned to the upper and lower cells. 
\begin{figure*}[t]
  % \rtron{fit all four subfigures into one single row} 
  
  \centering
  \subfloat[Shared landmarks, converging to a point]{\label{fig:stationary_share}{\includegraphics[height=4cm]{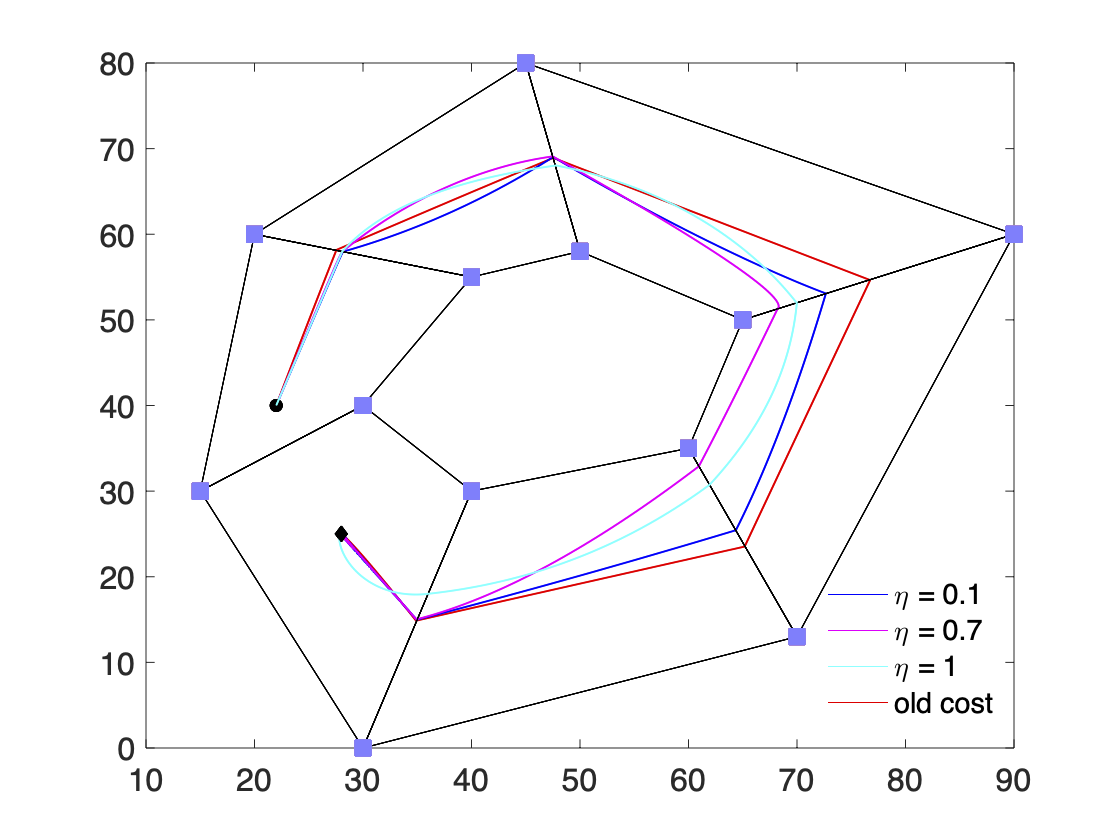}}}
  \subfloat[Two set of landmarks, converging to a point]{{\label{fig:stationary_diff}\includegraphics[height=4cm]{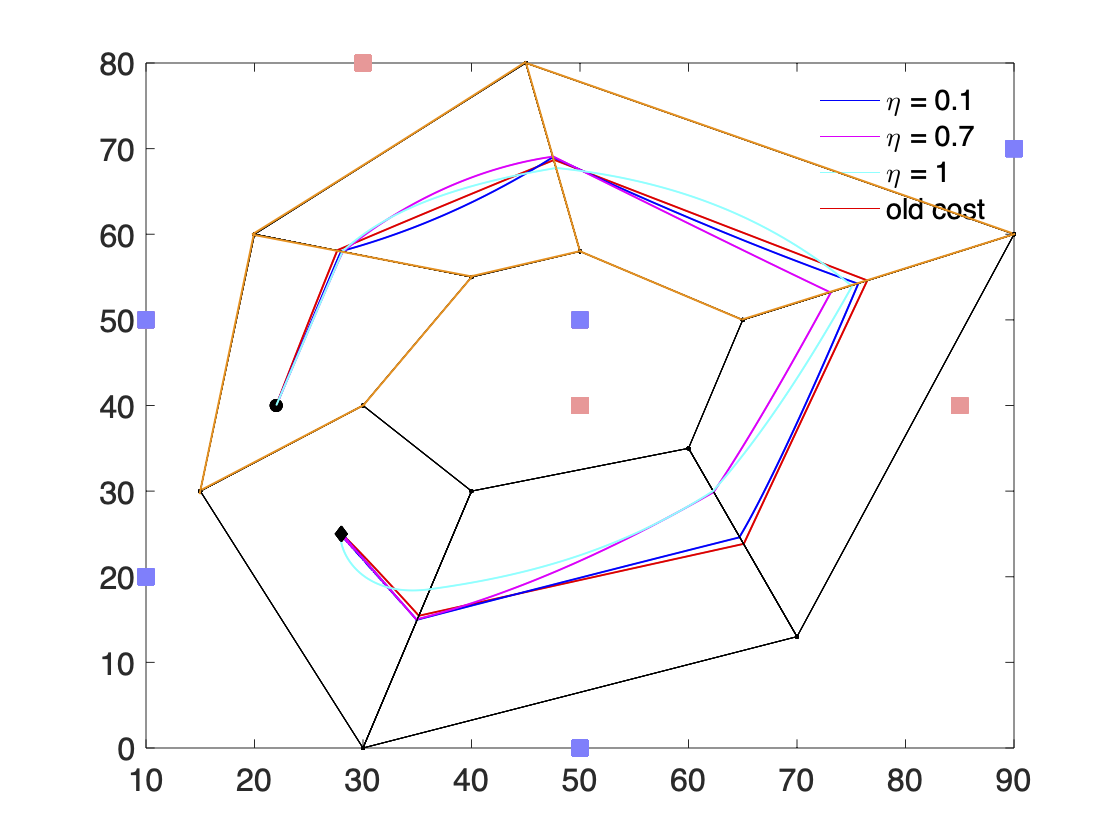} }}%
  \caption{Converging to a point while changing variable $\eta$. The black circle shows the starting point, while the black diamond represents the converging point. In Fig.~\ref{fig:stationary_share}, all cells share the same landmarks shown by blue squares. In Fig.~\ref{fig:stationary_diff}, Upper cells with orange edges use the orange landmarks, and lower cells with black edges use blue landmarks. For both cases, increasing variable $\eta$ makes the path smoother.}
  \label{fig:O1}
\end{figure*}

\begin{figure*}[t]
  \centering
  \subfloat[Shared landmarks,traversing]{\label{fig:o2-shared}{\includegraphics[height=4cm]{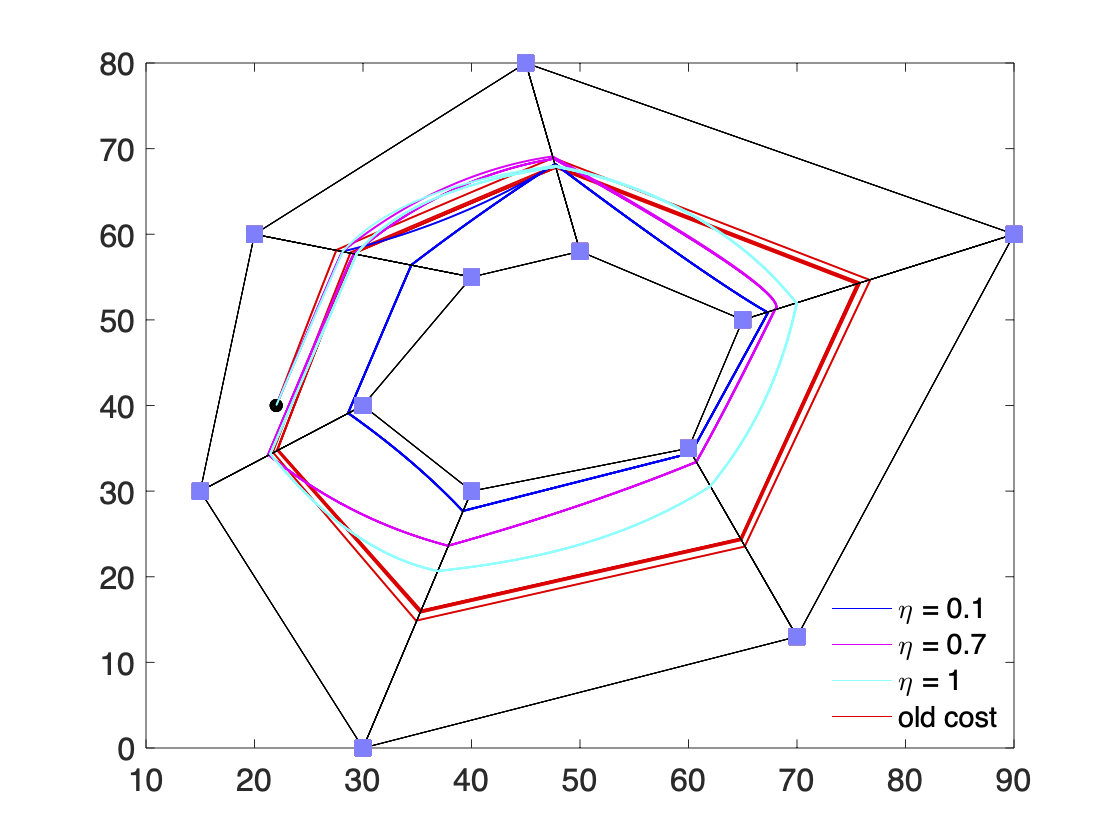}}}
  \subfloat[Two set of landmarks, traversing]{{\label{fig:o2-diff}\includegraphics[height=4cm]{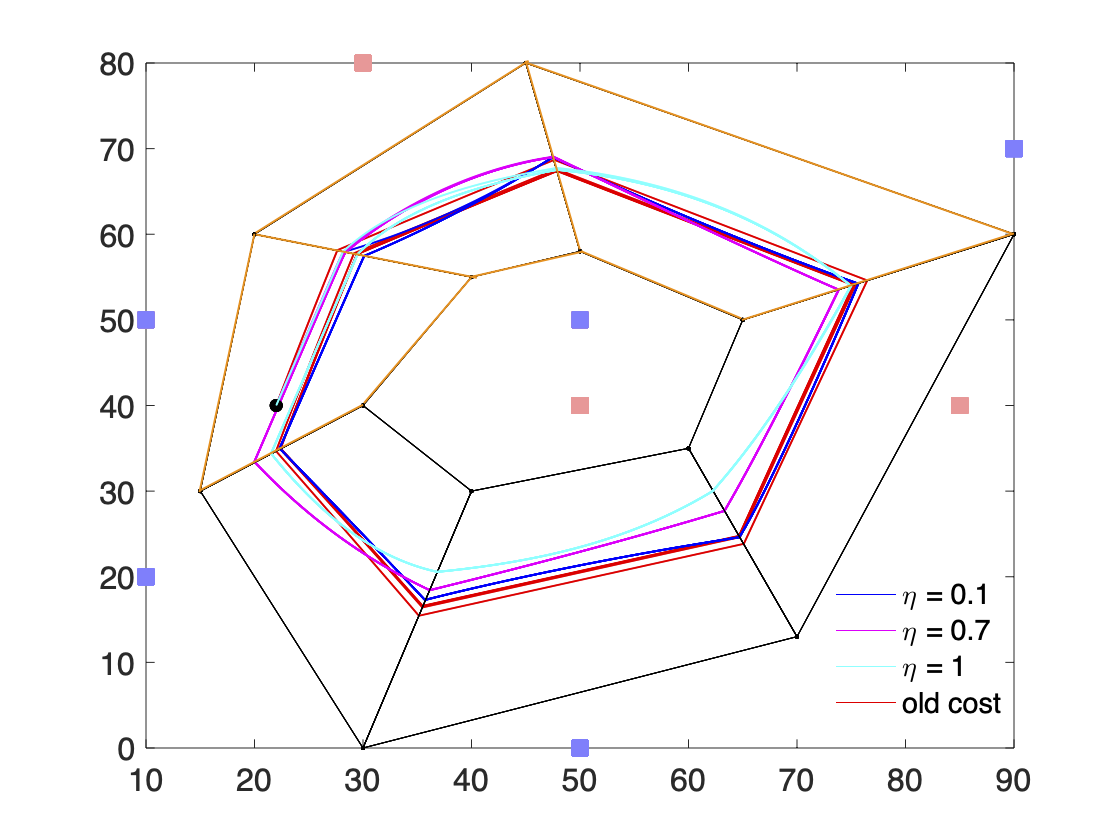} }}%
  \caption{Making loops changing variable $\eta$. The black circle shows the starting point. In Fig.~\ref{fig:o2-shared}, all cells share the same set of landmarks which are shown by blue squares. In Fig.~\ref{fig:o2-diff}, Upper cells with orange edges use the orange landmarks, and lower cells with black edges use blue landmarks. For both cases, increasing variable $\eta$ makes the path smoother.}
  \label{fig:O2}
\end{figure*}

%%%%%%%%%%%%%%%%%%%%%%%%%%%%%%%%%%%%%%%%%%%%%%%%%%%%%%%%%%%%%%%%%%%%%%%%%%%%%%%%

%%%%%%%%%%%%%%%%%%%%%%%%%%%%%%%%%%%%%%%%%%%%%%%%%%%%%%%%%
\section{CONCLUSIONS}\label{sec:conclusion}
In this work, we proposed a novel approach to synthesize a set of output feedback controllers on a cell decomposition of the environment; such decomposition is generated by a simplified version of the sampling-based \rrtstar{} method. We build a robust output feedback controller for each cell; the controller takes inputs on the relative displacements between a set of landmarks positions and the robot. The controllers for all cells are found simultaneously as the solution of a robust min-max Linear Program. The optimization includes CLF and CBF constraints to guarantee the stability and safety of the system and a new regularization term to smooth the transitions between consecutive cells of the environment. 
%and study the degree of smoothness of the generated path. 
In addition, we discuss strategies for handling practical problems deriving from the use of monocular cameras, such as limited fields of view (in which case the controller can be re-parameterized without solving a new optimization problem) and bearing measurements with unknown depths (in which case we propose a new triangulation approach to scale the bearing measurements before using them in the output feedback controller). 
We test the proposed algorithm in simulations to evaluate the performance of our approach under different measurements and the influence of the regularization term on the shape of the final path.
%%%%%%%%%%%%%%%%%%%%%%%%%%%%%%%%%%%%%%%%%%%%%%%%%%%%%%%%%%%%%%%%%%%%%%%%%%%%%%%%
\bibliographystyle{IEEEtran}
\bibliography{arixv_mahroo.bib}

\end{document}